\documentclass{lmcs}
\pdfoutput=1

\usepackage{lastpage}
\lmcsdoi{17}{3}{21}
\lmcsheading{}{\pageref{LastPage}}{}{}%
{Jun.~11,~2020}{Sep.~01,~2021}{}

\usepackage{logique_lmcs}
\usepackage{hyperref}

\usepackage[utf8]{inputenc}

\usepackage{amssymb}

\newcommand{\succcc}{\succ\hspace{-1em}\succ}

\newcommand{\HH}{\mathbb{H}}
\newcommand{\NN}{\mathbb{N}}

\newcommand{\RR}{\mathbb{R}}

\newcommand{\BBB}{\mbox{\sffamily B}}
\newcommand{\BBb}{\mbox{\bfseries\sffamily B}}
\newcommand{\CCC}{\mbox{\sffamily C}}
\newcommand{\CCb}{\mbox{\bfseries\sffamily C}}
\newcommand{\CF}{\mathfrak{C}}

\newcommand{\pF}{\mathfrak{p}}

\newcommand{\III}{\mbox{\sffamily I}}
\newcommand{\IIb}{\mbox{\bfseries\sffamily I}}
\newcommand{\HHH}{\mbox{\sffamily H}}
\newcommand{\KKK}{\mbox{\sffamily K}}
\newcommand{\KKb}{\mbox{\bfseries\sffamily K}}

\newcommand{\WWW}{\mbox{\sffamily W}}
\newcommand{\WWb}{\mbox{\bfseries\sffamily W}}
\newcommand{\VV}{\mbox{\sffamily\bf V}}

\newcommand{\ccb}{\mbox{\bfseries\sffamily cc}}
\newcommand{\aaa}{\mbox{\sffamily a}}

\newcommand{\ee}{\mbox{\rm\sffamily e}}

\newcommand{\hh}{\mbox{\sffamily h}}

\newcommand{\pp}{\mbox{\sffamily p}}

\newcommand{\kkk}{\mbox{\bf\sffamily k}}
\newcommand{\indi}{^{\mbox{\footnotesize int}}}

\newcommand{\PL}{\mbox{\sffamily PL}}

\newcommand{\forcec}{\raisebox{-0.2ex}{\small$\;\Vdash\,$}}
\newcommand{\nforcec}{\raisebox{-0.2ex}{\small$\;\nVdash\,$}}

\newcommand{\et}{{\scriptstyle\land}}

\newcommand{\lbr}{\langle}
\newcommand{\rbr}{\rangle}
\newcommand{\gl}{\gimel}

\newcommand{\ZFe}{ZF$_\varepsilon$}

\newcommand{\hto}{\hookrightarrow}

\newcommand{\llp}{(\hspace{-0.25em}(\hspace{-0.25em}(}
\newcommand{\rrp}{)\hspace{-0.25em})\hspace{-0.25em})}

\newcommand{\cerclefin}[1]{\textcircled{\raisebox{-0.6pt}{\footnotesize #1}}}
\newcommand{\cercle}[1]{\raisebox{0.7pt}{\cerclefin{#1}\hspace{-1.08em}\cerclefin{#1}}}

\begin{document}
\title{A program for the full axiom of choice}
\author{Jean-Louis Krivine}
\address{Université de Paris - C.N.R.S.}
\email{krivine@irif.fr}

\begin{abstract}
The theory of classical realizability is a framework for the Curry-Howard correspondence which enables to associate a program with each proof in Zermelo-Fraenkel set theory. But, almost all the applications of mathematics in physics, probability, statistics, etc.\ use Analysis i.e.\ the axiom of dependent choice (DC) or even the (full) axiom of choice (AC). It is therefore important to find explicit programs for these axioms. Various solutions have been found for DC, for instance the lambda-term called "bar recursion" or the instruction "quote" of LISP. We present here the first program for AC.
\end{abstract}

\maketitle

\section*{Introduction}
The Curry-Howard correspondence  enables to associate a program with each proof in classical natural deduction. But mathematical proofs not only use the rules of natural deduction, but also \emph{axioms}, essentially those of Zermelo-Fr\ae nkel set theory with axiom of choice. In order to transform these proofs into programs, we must therefore associate with each of these axioms suitable instructions, which is far from obvious.

The theory of \emph{classical realizability (c.r.)} solves this problem for all axioms of ZF\/, by means of a very rudimentary programming language:  the $\lbd_c$-calculus, that is to say the $\lbd$-calculus with a control instruction~\cite{griffin}.

The programs obtained in this way can therefore be written in practically any programming language. They are said to \emph{realize} the axioms of ZF\/.

But, almost all the applications of mathematics in physics, probability, statistics, etc.\ use \emph{Analysis}, that is to say \emph{the axiom of dependent choice}. The first program for this axiom, known since 1998 \cite{bbc}, is a pure $\lbd$-term called \emph{bar recursion}~\cite{spec,BO,stre}. In fact, c.r.\ shows that it provides also a program for
the \emph{continuum hypothesis}~\cite{kri5}.

Nevertheless, this solution requires the programming language to be limited to $\lbd_c$-calculus, prohibiting any other instruction, which is a severe restriction.

Classical realizability provides other programs for this axiom~\cite{kri1,kri2,kri3}, which require an additional
instruction (clock, signature, \ldots\,, or as in this paper, introduction of \emph{fresh variables}). On the other hand,
this kind of solution is very flexible regarding the programming language.

There remained, however, the problem of the \emph{full} axiom of choice. It is solved here, by means of new instructions
which allow \emph{branching} or \emph{parallelism}\footnote{A detailed study of the effect of such instructions on the
characteristic Boolean algebra $\gl2$ is made in~\cite{geoffroy}.}.
Admittedly, the program obtained in this way is rather complicated and we may hope for new solutions as simple
as bar recursion or clock. However, it already shows how we can turn all proofs of ZFC into programs
(more precisions about this in Remark~\ref{finrem} at the end of the paper).
Note that \emph{the proof is constructive} i.e.\ it gives explicitly a program for AC: you ''merely'' need to formalize this proof itself in natural deduction and apply the proof-program correspondence
to the result.

\subsection*{Outline of the paper}
The framework of this article is the \emph{theory of classical realizability}, which is explained in detail
in~\cite{kri1, kri2, kri3}. For the sake of simplicity, we will often refer to these papers for definitions and standard notations.
\begin{description}
\item[Section 1] Axioms and properties of \emph{realizability algebras (r.a.)} which are first order structures, a generalization of both combinatory algebras and ordered sets of forcing conditions.
Intuitively, their elements are programs.
\item[Section 2] Each r.a.\ is associated with a \emph{realizability model (r.m.)},
a generalization of the generic model in the theory of forcing. This model satisfies a set theory \ZFe\
which is a conservative extension of ZF\/, with a strong non extensional membership relation $\veps$.
\item[Section 3] Generic extensions of realizability algebras and models. We explain how a generic extension of a given
r.m.\ can be handled by means of a suitable extension of the corresponding r.a. A little technical because it mixes forcing and realizability.
\item[Section 4] Definition and properties of a particular realizability algebra
$\mathfrak{A}_0$. It contains an \emph{instruction of parallelism} and the programming language allows the introduction of many other instructions. This is therefore, in fact, a class of algebras.

At this stage, we obtain a program for the \emph {axiom of well ordered choice (WOC): the product of a family of non-empty sets whose index set is well ordered is non-empty.}

We get, at the same time, a new proof that WOC is weaker than AC (joint work with Laura Fontanella;
cf.~\cite{TJAC} for the usual proof).
\item[Section 5] Construction of a generic extension $\mathfrak{A}_1$ of the algebra
$\mathfrak{A}_0$ which allows to \emph{realize the axiom of choice: every set can be well ordered.}
\end{description}
Sections~\ref{e_ge} and~\ref{alg_A0} are independent.

Thanks to Asaf Karagila for several fruitful discussions \cite{kara}. He observed that, by exactly the same method,
we obtain a program, not only for AC, but for every axiom which can be shown compatible with ZFC \emph{by forcing}
(Theorem~\ref{theta_AC+A}); for instance CH, $2^{\aleph_0}=\aleph_2$ + Martin's axiom, etc.\ (but this does not apply to V = L).
\nocite{*}

\section{Realizability algebras (r.a.)}\label{sect_ra}
It is a first order structure, which is defined in~\cite{kri2}. We recall here briefly this definition and some
essential properties.

A \emph{realizability algebra} ${\mathcal A}$ is made up of three sets:
$\Lbd$ (the set of \emph{terms}), $\Pi$ (the set of \emph{stacks}),
$\Lbd\star\Pi$ (the set of \emph{processes}) with the following operations:

\begin{itemize}
\item[]$(\xi,\eta)\mapsto(\xi)\eta$ from $\Lbd^2$ into $\Lbd$ (\emph{application});
\item[]$(\xi,\pi)\mapsto\xi\ps\pi$ from $\Lbd\fois\Pi$ into $\Pi$ (\emph{push});
\item[]$(\xi,\pi)\mapsto\xi\star\pi$ from $\Lbd\fois\Pi$ into $\Lbd\star\Pi$ (\emph{process});
\item[]$\pi\mapsto\kk_\pi$ from $\Pi$ into $\Lbd$ (\emph{continuation}).
\end{itemize}
There are, in $\Lbd$, distinguished elements $\BBB,\CCC,\III,\KKK,\WWW,\Ccc$, called
\emph{elementary combinators} or \emph{instructions}.

{\bfseries Notation.} The term $(\ldots(((\xi)\eta_1)\eta_2)\ldots)\eta_n$ will be also written
$(\xi)\eta_1\eta_2\ldots\eta_n$ or even $\xi\eta_1\eta_2\ldots\eta_n$.
For instance: \ $\xi\eta\zeta=(\xi)\eta\zeta=(\xi\eta)\zeta=((\xi)\eta)\zeta$. 

We define a preorder on $\Lbd\star\Pi$, denoted by $\succ$, which is called \emph{execution};

$\xi\star\pi\succ\xi'\star\pi'$ is read as:
\emph{the process $\xi\star\pi$ reduces to $\xi'\star\pi'$.}

It is the smallest reflexive and transitive binary relation, such that, for any $\xi,\eta,\zeta\in\Lbd$ and $\pi,\varpi\in\Pi$, we have:

\begin{itemize}
\item[]$(\xi)\eta\star\pi\succ\xi\star\eta\ps\pi$ (\emph{push}).
\item[]$\III\star\xi\ps\pi\succ\xi\star\pi$ (\emph{no operation}).
\item[]$\KKK\star\xi\ps\eta\ps\pi\succ\xi\star\pi$ (\emph{delete}).
\item[]$\WWW\star\xi\ps\eta\ps\pi\succ\xi\star\eta\ps\eta\ps\pi$ (\emph{copy}).
\item[]$\CCC\star\xi\ps\eta\ps\zeta\ps\pi\succ\xi\star\zeta\ps\eta\ps\pi$ (\emph{switch}).
\item[]$\BBB\star\xi\ps\eta\ps\zeta\ps\pi\succ\xi\star(\eta)\zeta\ps\pi$ (\emph{apply}).
\item[]$\Ccc\star\xi\ps\pi\succ\xi\star\kk_\pi\ps\pi$ (\emph{save the stack}).
\item[]$\kk_\pi\star\xi\ps\varpi\succ\xi\star\pi$ (\emph{restore the stack}).
\end{itemize}

We are also given a subset $\bbot$ of $\Lbd\star\Pi$, called ''the pole'', such that:

\centerline{$\xi\star\pi\succ\xi'\star\pi'$, $\xi'\star\pi'\in\bbot$ \ $\Fl$ \ $\xi\star\pi\in\bbot$.}

Given two processes $\xi\star\pi,\xi'\star\pi'$, the notation $\xi\star\pi\succcc\xi'\star\pi'$ means:

\centerline{$\xi\star\pi\notin\bbot\Fl\xi'\star\pi'\notin\bbot$.}

Therefore, obviously, \ $\xi\star\pi\succ\xi'\star\pi'\;\Fl\;\xi\star\pi\succcc\xi'\star\pi'$.

Finally, we choose a set of terms \ \PL$_{\mathcal A}\subset\Lbd$, containing the elementary combinators:
$\BBB,\CCC,\III,\KKK,\WWW,\Ccc$ and closed by application. They are called the
\emph{proof-like terms of the algebra~${\mathcal A}$}. We write also \PL\ instead of
\PL$_{\mathcal A}$ if there is no ambiguity about ${\mathcal A}$.

The algebra ${\mathcal A}$ is called \emph{coherent} if, for every proof-like term \
$\theta\in\mbox{\PL}_{\mathcal A}$, there exists a stack $\pi$ such that $\theta\star\pi\notin\bbot$.

\begin{rem} A \emph{set of forcing conditions} can be considered as a degenerate case of realizability
algebra, if we present it in the following way:

An inf-semi-lattice $P$, with a greatest element $\1$ and an
initial segment $\bbot$ of $P$ (the set of \emph{false} conditions). Two conditions $p,q\in P$ are called
\emph{compatible} if their g.l.b.\ $p\et q$ is not in $\bbot$.

We get a realizability algebra if we set $\Lbd=\Pi=\Lbd\star\Pi=P$; $\BBB=\CCC=\III=\KKK=\WWW=\Ccc=\1$ and \PL$=\{\1\}$;
$(p)q=p\ps q=p\star q=p\et q$ and $\kk_p=p$. The preorder $p\succ q$ is defined as $p\le q$, i.e.\ $p\et q=p$. The condition of coherence is $\1\notin\bbot$.
\end{rem}

We can translate $\lbd$-terms into terms of $\Lbd$ built with the combinators $\BBB,\CCC,\III,\KKK,\WWW$
in such a way that weak head reduction is valid:

\centerline{$\lbd x\,t[x]\star u\ps\pi\succ t[u/x]\star\pi$}

where $\lbd x\,t[x],u$ are terms and $\pi$ is a stack.

This is done in~\cite{kri2}. Note that the usual $(\KKK$,{\sffamily S})-translation does not work.

In order to write programs, $\lbd$-calculus \nocite{howard} is much more intuitive than combinatory algebra, so that
we use it extensively in the following. But combinatory algebra \cite{curry} is better for theory, in particular
because it is a first order structure.

\section{Realizability models (r.m.)}\label{sect_rm}
The framework is very similar to that of forcing, which is anyway a particular case.

We use a first order language with three binary symbols $\neps,\notin,\subset$
($\eps$ is intended to be a strong, non extensional membership relation; $\in$ and $\subset$ have their usual meaning in ZF).

First order formulas are written with the only logical symbols $\to,\pt,\bot,\top$.

The symbols $\neg,\land,\lor,\dbfl,\ex$ are defined with them in the usual way.

Given a  realizability algebra, we get a \emph{realizability model (r.m.)} as follows:

We start with a model ${\mathcal M}$ of ZFC (or even ZFL) called the \emph{ground model}.
The axioms of ZFL are written with the sublanguage $\{\notin,\subset\}$.

We build a model ${\mathcal N}$ of a new set theory \ZFe, in the language
$\{\neps,\notin,\subset\}$, the axioms of which are given in~\cite{kri1}. We recall them below, using
the following rather standard abbreviations:

$F_1\to(F_2\to\ldots(F_n\to G)\ldots)$ is written $F_1,F_2,\ldots,F_n\to G$ or even $\vec{F}\to G$.

We use the notation $\ex x\{F_1,F_2,\ldots,F_n\}$ for $\pt x(F_1\to(F_2\to\cdots\to(F_n\to\bot)\cdots))\to\bot$.

Of course, $x\eps y$ and $x\in y$ are the formulas $x\neps y\to\bot$ and $x\notin y\to\bot$.

The notation $x=_\in\!y\to F$ means $x\subset y,y\subset x\to F$. Thus $x=_\in\!y$,
which represents the usual (extensional) equality of sets, is the pair of formulas $\{x\subset y,y\subset x\}$.

We use the notations:

$(\pt x\eps a)F(x)$ for $\pt x(\neg F(x)\to x\neps a)$ and
$(\ex x\eps a)\vec{F}(x)$ for $\neg\pt x(\vec{F}(x)\to x\neps a)$.

For instance,\ $(\ex x\eps y)(t=_\in\!u)$ is the formula
$\neg\pt x(t\subset u,u\subset t\to x\neps y)$.

The axioms of \ZFe\ are the following:

\begin{enumerate}
    \setcounter{enumi}{-1}
\item Extensionality axioms.

$\pt x\pt y(x\in y\dbfl(\ex z\eps y)x=_\in\!z)$;
$\pt x\pt y(x\subset y\dbfl(\pt z\eps x)z\in y)$.
\item Foundation scheme.

$\pt\vec{a}(\pt x((\pt y\eps x)F(y,\vec{a})\to F(x,\vec{a}))\to\pt x\,F(x,\vec{a}))$
for every formula $F(x,a_1,\ldots,a_n)$.
\end{enumerate}

The intuitive meaning of these axioms is that $\varepsilon$ is a well
founded relation and the relation~$\in$ is obtained by ``~collapsing~''
$\varepsilon$ into an extensional relation.

The following axioms essentially express that the relation $\eps$
satisfies the Zermelo-Fraenkel axioms {\em except extensionality}.
\begin{enumerate}
    \setcounter{enumi}{1}
\item Comprehension scheme.

$\pt\vec{a}\pt x\ex y\pt z(z\eps y\dbfl(z\eps x\land F(z,\vec{a}))$
for every formula $F(z,\vec{a})$.
\item Pairing axiom.

$\pt a\pt b\ex x\{a\eps x, b\eps x\}$.
\item Union axiom.

$\pt a\ex b(\pt x\eps a)(\pt y\eps x)\,y\eps b$.
\item Power set axiom.

$\pt a\ex b\pt x(\ex y\eps b)
\pt z(z\eps y\dbfl(z\eps a\land z\eps x))$.
\item Collection scheme.

$\pt\vec{a}\pt x\ex y(\pt u\eps x)(\ex v\,F(u,v,\vec{a})\to(\ex v\eps y)F(u,v,\vec{a}))$
for every formula $F(u,v,\vec{a})$.
\item Infinity scheme.

$\pt\vec{a}\pt x\ex y\{x\eps y,(\pt u\eps y)
(\ex v\,F(u,v,\vec{a})\to(\ex v\eps y)F(u,v,\vec{a}))\}$
for every formula $F(u,v,\vec{a})$.
\end{enumerate}

It is shown in~\cite{kri1} that \ZFe\ is a \emph{conservative extension of ZF\/.}

For each formula $F(\vec{a})$ of \ZFe\ (i.e.\ written with $\neps,\notin,\subset$) with parameters $\vec{a}$
in the ground model ${\mathcal M}$ we define, in ${\mathcal M}$, a \emph{falsity value} $\|F(\vec{a})\|$ which is
a subset of $\Pi$ and a \emph{truth value} $|F(\vec{a})|$ which is a subset of $\Lbd$.

The notation $t\force F(\vec{a})$ (read ``$t$ realizes $F(\vec{a})$'' or ``$t$ forces $F(\vec{a})$''
in the particular case of forcing) means $t\in|F(\vec{a})|$.

We set first $|F(\vec{a})|=\{t\in\Lbd\;;\;(\pt\pi\in\|F(\vec{a})\|)(t\star\pi\in\bbot)\}$ so that
we only need to define $\|F(\vec{a})\|$, which we do by induction on $F$:

\begin{enumerate}
\item Definition of $\|a\neps b\|$:

$\|a\neps b\|=\{\pi\in\Pi\;;\;(a,\pi)\in b\}$; $\|\bot\|=\Pi$; $\|\top\|=\vide$;
\item Definition of $\|a\subset b\|$ and $\|a\notin b\|$ by induction on
$(\mbox{rk}(a)\cup\mbox{rk}(b),\mbox{rk}(a)\cap\mbox{rk}(b))$:

$\|a\subset b\|=\bigcup_c\{t\ps\pi\in\Pi\;;\;t\force c\notin b,(c,\pi)\in a\}$;

$\|a\notin b\|=\bigcup_c\{t\ps u\ps\pi\in\Pi\;;\;t\force c\subset a,u\force a\subset c,
(c,\pi)\in b\}$.
\item Definition of $\|F\|$ for a non atomic formula $F$, by induction on the length:

$\|F\to F'\|=\{t\ps\pi\in\Pi\;;\;t\force F,\pi\in\|F'\|\}$;
$\|\pt x\,F(x)\|=\bigcup_a\|F(a)\|$.
\end{enumerate}

This notion of realizability has two essential properties given by theorems~\ref{adeq}
and~\ref{ZFe} below. They are proved in~~\cite{kri1}.
\begin{thm}[Adequacy lemma]\label{adeq}
$\force{}$ is compatible with classical natural deduction, i.e.:

If $t_1,\ldots,t_n,t$ are $\lbd_c$-terms such that $t_1:F_1,\ldots,t_n:F_n\vdash t:F$ in classical natural deduction,
then $t_1\force F_1,\ldots,t_n\force F_n\To t\force F$.

In particular, any valid formula is realized by a proof-like term.
\end{thm}
\begin{rem}\label{remadeq} The proof of Theorem~\ref{adeq} uses only item~3 in the above definition of $\|F\|$.
In other words, the values of $\|F\|$ for atomic formulas $F$ are arbitrary. This will be used in Sections~\ref{e_ge}
and~\ref{alg_A1}.
\end{rem}

\begin{thm}\label{ZFe}
The axioms of \ZFe\ are realized by proof-like terms.
\end{thm}
It follows that \emph{every closed formula which is consequence of \ZFe\ and, in particular,
every consequence of ZF\/, is realized by a proof-like term}.

In the following, we shall simply say ``the formula $F$ is realized'' instead of ``realized by a proof-like term''
and use the notation $\force F$.

Theorem~\ref{ZFe} is valid \emph{for every r.a.} The aim of this paper is to realize the
full axiom of choice AC in some particular r.a.\ \emph{suitable for programming}.
\begin{rem} Note that AC is realized for any r.a.\ associated with
a set of forcing conditions (generic extension of ${\mathcal M}$). But in this case, there is only one
proof-like term which is the greatest element~$\1$.
\end{rem}

We define the strong (Leibnitz) equality $a=b$ by $\pt z(a\neps z\to b\neps z)$. It is trivially
transitive and it is symmetric by comprehension. This equality satisfy the first order axioms of equality
$\pt x\pt y(x=y\to(F(x)\to F(y)))$ (by comprehension scheme of \ZFe) and is therefore the equality in
the r.m.\ ${\mathcal N}$. 

\begin{lem}\label{eq}
$\|a=b\|=\|\top\to\bot\|=\{\xi\ps\pi\;;\;\xi\in\Lbd,\pi\in\Pi\}$ if $a\ne b$;

$\|a=a\|=\|\bot\to\bot\|=\{\xi\ps\pi\;;\;\xi\in\Lbd,\pi\in\Pi,\xi\star\pi\in\bbot\}$.
\end{lem}
\begin{proof}
Let $z=\{b\}\fois\Pi$ so that $\|b\neps z\|=\Pi$. If $a\ne b$ then $\|a\neps z\|=\vide$ and therefore:

$\|a\neps z\to b\neps z\|=\|\top\to\bot\|$.

If $a=b$, then $\|a\neps z\|=\Pi$ and therefore $\|a=b\|=\|\bot\to\bot\|$.
\end{proof}
Finally, it is convenient to define first $\ne$ by $\|a\ne a\|=\|\bot\|=\Pi$; $\|a\ne b\|=\|\top\|=\vide$ if $a\ne b$;
and to define $a=b$ as $a\ne b\to\bot$.

We define a preorder $\le$ on the set ${\mathcal P}(\Pi)$ of ``falsity values'' by setting:

$X\le Y\Dbfl$ there exists a proof-like term $\theta\force X\to Y$. By Theorem~\ref{adeq}, we easily see \cite{kri1}
that $({\mathcal P}(\Pi),\le)$ is a Boolean algebra $\mathfrak{B}_{\mathcal A}$ \emph{if the r.a.\ ${\mathcal A}$ is coherent}.
Every formula $F(\vec{a})$ of \ZFe\ with parameters in the ground model ${\mathcal M}$ has a value
$\|F(\vec{a})\|$ in this Boolean algebra.

By means of any ultrafilter on $\mathfrak{B}_{\mathcal A}$, we thus obtain a complete consistent theory
in the language $\{\neps,\notin,\subset\}$ with parameters in ${\mathcal M}$. We take any model ${\mathcal N}$
of this theory and call it \emph{the realizability model (r.m.)} of the realizability algebra
${\mathcal A}$.

Therefore, ${\mathcal N}$ is a model of \ZFe, and in particular, a model of ZF\/, that we will
call ${\mathcal N}_\in$.

Thus ${\mathcal N}_\in$ is simply the model ${\mathcal N}$ restricted to the language $\{\notin,\subset\}$.

\begin{rem}
The ground model ${\mathcal M}$ is contained in ${\mathcal N}$ since every
element of it is a symbol of constant. But ${\mathcal M}$ is not a submodel of ${\mathcal N}$ for the common
language $\{\notin,\subset\}$; and, except in the case of forcing, not every element of ${\mathcal N}$
``has a name'' in ${\mathcal M}$.

When $F$ is a closed formula of \ZFe, the two assertions ${\mathcal N}\models F$ and $\force F$ have essentially the same
meaning, since ${\mathcal N}$ represents \emph{any r.m.} for the given r.a. But the second formulation requires a formal
proof.
\end{rem}

\subsection*{Functionals}
\emph{A functional relation defined in ${\mathcal N}$} is given by a formula $F(x,y)$ of \ZFe\ such that~~
${\mathcal N}\models\pt x\pt y\pt y'(F(x,y),F(x,y')\to y=y')$.

A \emph{function} is a functional relation which is a set.

We define now some special functional relations on ${\mathcal N}$ which we call \emph{functionals
defined in~${\mathcal M}$} or \emph{functional symbols}:

For each functional relation $f:{\mathcal M}^k\to{\mathcal M}$ defined in the ground model ${\mathcal M}$,
we add the functional symbol $f$ to the language of \ZFe. The application of $f$ to an argument $a$ will be
denoted $f[a]$. Therefore $f$ is also defined in ${\mathcal N}$.

We call this (trivial) operation \emph{the extension to ${\mathcal N}$ of the functional $f$} defined in the ground model. It is a fundamental tool in all that follows.

Note the use of brackets for $f[a]$ in this case.

\begin{thm}\label{idterm}
Let $t_1,u_1\ldots,t_n,u_n,t,u$ be $k$-ary terms built with functional symbols, such that \
${\mathcal M}\models\pt\vec{x}(t_1[\vec{x}]=u_1[\vec{x}],\ldots,t_n[\vec{x}]=u_n[\vec{x}]\to t[\vec{x}]=u[\vec{x}])$.

Then $\lbd x_1\ldots\lbd x_n\lbd x(x_1)\ldots(x_n)x
\force\pt\vec{x}(t_1[\vec{x}]=u_1[\vec{x}],\ldots,t_n[\vec{x}]=u_n[\vec{x}]\to t[\vec{x}]=u[\vec{x}])$.
\end{thm}
\begin{proof}
This easily follows from the definition above of $\|a=b\|$.
\end{proof}
As a first example let the unary functional $\Phi_F[X]$ be defined in ${\mathcal M}$ by:

\centerline{$\Phi_F[X]=\{(x,\xi\ps\pi)\;;\;\xi\force F(x),(x,\pi)\in X\}$.}

We shall denote it by $\{x\eps X\;;\;F(x)\}$ (in this notation, $x$ is a bound variable) because it corresponds
to the comprehension scheme in the model ${\mathcal N}$. Note that the use of $\veps$ reminds that this expression
must be interpreted in ${\mathcal N}$.

We define now in ${\mathcal M}$ the unary functional $\gl X=X\fois\Pi$, so that we have:

$\|x\neps\gl X\|=\Pi$ if $x\in X$ and $\|x\neps\gl X\|=\vide$ if $x\notin X$.

For any $X$ in ${\mathcal M}$, we define the quantifier $\pt x^{\gl X}$ by setting
$\|\pt x^{\gl X}F(x)\|=\bigcup_{x\in X}\|F(x)\|$.

\begin{lem}\label{glX}
$\force\pt x^{\gl X}F(x)\dbfl\pt x(x\eps\gl X\to F(x))$.
\end{lem}
\begin{proof}
In fact, we have $\|\pt x(\neg F(x)\to x\neps\gl X)\|=\|\pt x^{\gl X}\neg\neg F(x)\|$.

Now we have trivially:

$\lbd x(x)\III\,\force\pt x^{\gl X}F(x)\to\pt x^{\gl X}\neg\neg F(x)$
and $\Ccc\force\pt x^{\gl X}\neg\neg F(x)\to\pt x^{\gl X}F(x)$.
\end{proof}
\begin{lem}\label{fun-gl}
Let $f$ be a functional $k$-ary symbol defined in ${\mathcal M}$ such that $f:X_1\fois\cdots\fois X_k\to X$.
Then its extension to ${\mathcal N}$ is such that $f:\gl X_1\fois\cdots\fois\gl X_k\to\gl X$.
\end{lem}
\begin{proof}
Trivial.
\end{proof}
By Theorem~\ref{idterm} and Lemma~\ref{fun-gl}, the algebra operations on the Boolean algebra $2=\{0,1\}$
extended to $\gl2$, turn it into a Boolean algebra which we call \emph{the characteristic Boolean algebra}
of the r.m.\ ${\mathcal N}$. In the ground model ${\mathcal M}$, we define the functional $(a,x)\mapsto ax$
from $2\fois{\mathcal M}$ into ${\mathcal M}$ by $0x=\vide$ and $1x=x$. It extends to ${\mathcal N}$ into a functional
$\gl2\fois{\mathcal N}\to{\mathcal N}$ such that $(ab)x=a(bx)$ for $a,b\eps\gl2$ and every $x$ in ${\mathcal N}$.

\begin{lem}\label{fax-afx}
$\III\,\force\pt\vec{x}\pt\vec{y}\pt a^{\gl2}(af[\vec{x},\vec{y}]=af[a\vec{x},\vec{y}])$ for every functional symbol
$f$ defined in~${\mathcal M}$.
\end{lem}
\begin{proof}
Immediate since $\xi\force\pt a^{\gl2}F(a)$ means $(\xi\force F(0))\land(\xi\force F(1))$.
\end{proof}
For any formula $F(\vec{x})$ of ZF we define, in ${\mathcal M}$, a functional $\lbr F(\vec{x})\rbr$
with value in $\{0,1\}$ which is the truth value of this formula in ${\mathcal M}$.\footnote{The formal definition in ZF
is: $\pt\vec{x}((F(\vec{x})\to\lbr F(\vec{x})\rbr=1)\land(\neg F(\vec{x})\to\lbr F(\vec{x})\rbr=0))$.}
The extension of this functional to the model ${\mathcal N}$ takes its values in the Boolean algebra
$\gl2$~(cf.~\cite{kri4}).

The binary functionals $\lbr x\notin y\rbr$ and $\lbr x\subset y\rbr$ define on the r.m.\ ${\mathcal N}$
a structure of \emph{Boolean model} on the Boolean algebra $\gl2$, that we denote by ${\mathcal M}_{\gl2}$.
It is an elementary extension of~${\mathcal M}$ since the truth value of every closed formula of ZF with
parameters in ${\mathcal M}$ is the same in~${\mathcal M}$ and~${\mathcal M}_{\gl2}$.

Any ultrafilter ${\mathcal U}$ on $\gl2$ would therefore give a (two-valued) model ${\mathcal M}_{\gl2}/{\mathcal U}$ which is an
elementary extension of ${\mathcal M}$. In~\cite{kri4}, it is shown that there exists one and only one
ultrafilter ${\mathcal D}$ on~$\gl2$ such that the model ${\mathcal M}_{\gl2}/{\mathcal D}$, which we shall denote
as ${\mathcal M}_{\mathcal D}$, is well founded (in ${\mathcal N}$). The binary relations $\notin,\subset$ of
${\mathcal M}_{\mathcal D}$ are thus defined by $\lbr x\notin y\rbr\eps{\mathcal D}$ and
$\lbr x\subset y\rbr\eps{\mathcal D}$.

Moreover, ${\mathcal M}_{\mathcal D}$ is isomorphic with a transitive submodel of ${\mathcal N}_\in$ with the
same ordinals. In fact, if we start with a ground model ${\mathcal M}$ which satisfies V = L, then
${\mathcal M}_{\mathcal D}$ is isomorphic with the constructible class of ${\mathcal N}_\in$.

\begin{rem}\label{mod_N} We have defined \emph{four first order structures} on the model ${\mathcal N}$:
\begin{itemize}
\item The realizability model ${\mathcal N}$ itself uses the language $\{\neps,\notin,\subset\}$ of \ZFe.

The equality on ${\mathcal N}$ is the Leibnitz equality =, which is the strongest possible.
\item The model ${\mathcal N}_\in$ of ZF is restricted to the language $\{\notin,\subset\}$.

The equality on ${\mathcal N}_\in$ is the extensional equality~~$=_\in$.
\item The Boolean model ${\mathcal M}_{\gl2}$ with the language $\{\notin,\subset\}$  of ZF and with truth
values in $\gl2$; it is an elementary extension of the ground model ${\mathcal M}$.

The equality on ${\mathcal M}_{\gl2}$
is $\lbr x=y\rbr=1$ which is the same as Leibnitz equality.
\item The model ${\mathcal M}_{\mathcal D}$ with the same language, also an elementary extension of ${\mathcal M}$;
if $F(\vec{a})$ is a closed formula of ZF with parameters (in ${\mathcal N}$), then 
${\mathcal M}_{\mathcal D}\models F(\vec{a})$ iff ${\mathcal N}\models\lbr F(\vec{a})\rbr\eps{\mathcal D}$.

The equality on ${\mathcal M}_{\mathcal D}$ is given by $\lbr x=y\rbr\eps{\mathcal D}$.
\end{itemize}
\end{rem}

The proof of existence of the ultrafilter ${\mathcal D}$ in~\cite{kri4} is not so simple. But it is
useless in the present paper, because $\gl2$ will be the four elements algebra, with two atoms $a_0,a_1$
which give the two trivial ultrafilters on $\gl2$. It is easily seen (Lemma~\ref{wellf}) that one of them,
say~$a_0$ gives a well founded model denoted by~${\mathcal M}_{a_0}$ which is the class $a_0{\mathcal N}={\mathcal M}_{\mathcal D}$.
The class ${\mathcal M}_{a_1}=a_1{\mathcal N}$ is also an elementary extension of ${\mathcal M}$ (but not well founded,
cf.~Remark~\ref{notwf}).
Finally we have ${\mathcal M}_{\gl2}={\mathcal N}=a_0{\mathcal N}\fois a_1{\mathcal N}$ since the Boolean model ${\mathcal M}_{\gl2}$
is simply a product in this case, and equality is the same on ${\mathcal M}_{\gl2}$ and ${\mathcal N}$ (Remark~\ref{mod_N}).

\subsection*{Integers}
We define, in the ground model ${\mathcal M}$, the functional $x\mapsto x^+=x\cup\{x\}$ and extend it to the r.m.~${\mathcal N}$.
It is injective in ${\mathcal M}$ and therefore also in ${\mathcal N}$.

For each $n\in\NN$, we define $\ul{n}\in$ \PL\ by induction: $\ul{0}=\lbd x\lbd y\,y=\KKK\III$; $\ul{n}^+=s\ul{n}$
where $s=\lbd n\lbd f\lbd x(nf)(f)x=(\BBB\WWW)(\BBB)\BBB$ (and $n^+$ is $n+1$).

We define  $\wt{\NN}=\{(n,\ul{n}\ps\pi)\;;\;n\in\NN,\pi\in\Pi\}$. We can use it as  the set of integers of the model~${\mathcal N}$
as shown by the following Theorem~\ref{entiers}.

We define the quantifier $\pt n\indi$ by setting $\|\pt n\indi F(n)\|=\{\ul{n}\ps\pi\;;\;n\in\NN,\pi\in\|F(n)\|\}$.

\begin{thm}\label{entiers}
For every formula $F(x)$ of \ZFe, the following formulas are realized:

i)~~$0\eps\wt{\NN}$; $\pt n(n\eps\wt{\NN}\to n^+\eps\wt{\NN})$;

ii)~~$F(0),\pt n(F(n)\to F(n^+))\to(\pt n\eps\wt{\NN})F(n)$;

iii)~~$\pt n\indi F(n)\dbfl(\pt n\eps\wt{\NN})F(n)$.
\end{thm}
\begin{proof}
i)~~Let $\xi\force0\neps\wt{\NN}$; then $\xi\star\ul{0}\ps\pi\in\bbot$ for all $\pi\in\Pi$.
Therefore $\lbd x(x)\ul{0}\force0\neps\wt{\NN}\to\bot$.

Let $\xi\force n^+\neps\wt{\NN}$ and $\ul{n}\ps\pi\in\|n\neps\wt{\NN}\|$.
We have $\xi\star s\ul{n}\ps\pi\in\bbot$ and therefore $\theta\star\xi\ps\ul{n}\ps\pi\in\bbot$ with
$\theta=\lbd x\lbd n(x)(s)n$. Thus $\theta\force\pt n(n^+\neps\wt{\NN}\to n\neps\wt{\NN})$.

ii)~~We show $\lbd x\lbd y\lbd z\lbd n(x)(n)yz\force\neg F(0),\pt n(F(n^+)\to F(n))\to\pt m(F(m)\to m\neps\wt{\NN})$.

Let $\xi\force\neg F(0),\eta\force\pt n(F(n^+)\to F(n)),\zeta\force F(m)$ and $\ul{m}\ps\pi\in\|m\neps\wt{\NN}\|$.

Let us show that $\ul{m}\eta\zeta\force F(0)$ by induction on $m$:

This is clear if $m=0$. Now, $\ul{m}^+\eta\zeta=s\ul{m}\eta\zeta\succ(m\eta)(\eta)\zeta$ and $\eta\zeta\force F(m)$
since $\zeta\force F(m^+)$. Therefore $(m\eta)(\eta)\zeta\force F(0)$ by the induction hypothesis.

It follows that $(\xi)(\ul{m})\eta\zeta\force\bot$, hence the result.

iii)~~Let us use Lemma~\ref{neg}. We have:

$\|\pt n(^\neg F(n)\to n\neps\wt{\NN})\|=
\{\kk_\pi\ps\ul{n}\ps\varpi\;;\;n\in\NN,\pi\in\|F(n)\|,\varpi\in\Pi\}$ and by definition:

$\|\pt n\indi F(n)\|=\{\ul{n}\ps\pi\;;\;n\in\NN,\pi\in\|F(n)\|\}$. It follows easily that:

$\lbd x\lbd n(\Ccc)\lbd k(x)kn\force\pt n(^\neg F(n)\to n\neps\wt{\NN})\to\pt n\indi F(n)$

$\lbd x\lbd k\lbd n(k)(x)n\force\pt n\indi F(n)\to\pt n(^\neg F(n)\to n\neps\wt{\NN})$.
\end{proof}

\subsection*{Some useful notations}
For every set of terms $X\subset\Lbd$ and every closed formula $F$ we can define an ``extended formula''
$X\to F$ by setting $\|X\to F\|=\{\xi\ps\pi\;;\;\xi\in X,\pi\in\|F\|\}$.

For instance, for every formula $F$, we define $^\neg F=\{\kk_\pi\;;\;\pi\in\|F\|\}$. It is a useful
equivalent of $\neg F$ by the following:

\begin{lem}\label{neg}
$\force{}^\neg F\dbfl\neg F$.
\end{lem}
\begin{proof}
If $\pi\in\|F\|$, then $\kk_\pi\force\neg F$ and therefore $\III\,\force^\neg F\to\neg F$.

Conversely, if $\xi\force^\neg F\to\bot$, then $\xi\star\kk_\pi\ps\pi\in\bbot$ for every $\pi\in\|F\|$;

thus $\Ccc\force\neg^\neg F\to F$.
\end{proof}
If $t,u$ are terms of the language of ZF, built with functionals in ${\mathcal M}$,
we define another ``extended formula'' $t=u\hto F$ by setting:

\centerline{$\|t=u\hto F\|=\vide$ if $t\ne u$; $\|t=u\hto F\|=\|F\|$ if $t=u$.}

We write $(t_1=u_1),\ldots,(t_n=u_n)\hto F$ for $(t_1=u_1)\hto(\cdots\hto((t_n=u_n)\hto F)\cdots)$.

\begin{lem}
$\force((t=u\hto F)\dbfl(t=u\to F))$.
\end{lem}
\begin{proof}
We have immediately \ $\III\,\force\neg F,(t=u\hto F)\to t\ne u$.

Conversely $\lbd x(x)\,\III\force(t=u\to F)\to(t=u\hto F)$.
\end{proof}
For instance, the conclusion of Theorem~\ref{idterm} may be rewritten as:

$\III\,\force\pt\vec{x}(t_1[\vec{x}]=u_1[\vec{x}],\ldots,t_n[\vec{x}]=u_n[\vec{x}]\hto t[\vec{x}]=u[\vec{x}])$.

Lemmas~\ref{eps-Cl}, \ref {ax-in-aX-gea}, \ref{choix_pi} and Theorem~\ref{V_alpha} below will be used in the following sections.

\begin{lem}\label{eps-Cl}
$\force\pt x\pt y(x\eps y\to\lbr x\in\mbox{\rm Cl}[y]\rbr=1)$.
\end{lem}
\begin{proof}
In the model ${\mathcal M}$, the unary functional symbol Cl denotes the \emph{transitive closure}.

We show \ $\III\,\force\pt x\pt y(\lbr x\in\mbox{Cl}[y]\rbr\ne1\to x\neps y)$:
let $\xi\force\lbr x\in\mbox{Cl}[y]\rbr\ne1$ and $\pi\in\|x\neps y\|$.

Then $(x,\pi)\in y$, therefore $x\in\mbox{Cl}(y)$. It follows that $\xi\force\bot$.
\end{proof}

\begin{lem}\label{ax-in-aX-gea}
$\III\,\force\pt X\pt a^{\gl2}\pt x(\lbr ax\in aX\rbr\ge a\hto ax\eps\Phi[X])$ where $\Phi$ is the functional symbol
defined in ${\mathcal M}$ by $\,\Phi[X]=(X\cup\{0\})\fois\Pi$ i.e.\ $\gl(X\cup\{0\})$.
\end{lem}
The notation $b\ge a$ for $a,b\eps\gl2$ means, of course, $ab=a$.

\begin{proof}
This amounts to show:

  \begin{enumerate}
    \item $x\in X\To\,\III\force x\neps\Phi[X]\to\bot$;
    \item $\III\force 0\neps\Phi[X]\to\bot$.
  \end{enumerate}

Both are trivial.
\end{proof}

\begin{lem}\label{choix_pi}
Let $F(x,\vec{y})$ be a formula in \ZFe. Then:

\centerline{$\III\,\force\pt\vec{y}(\pt\varpi^{\gl\Pi}F(f[\varpi,\vec{y}],\vec{y})\to\pt x\,F(x,\vec{y}))$}

for some functional symbol $f$ defined in ${\mathcal M}$.
\end{lem}
\begin{proof}
Let $\vec{a}=(a_1,\ldots,a_k)$ in ${\mathcal M}$. Then, we have:

$\pi\in\|\pt x\,F(x,\vec{a})\|\Dbfl{\mathcal M}\models\ex x(\pi\in\|F(x,\vec{a})\|)
\Dbfl{\mathcal M}\models\pi\in\|F(f[\pi,\vec{a}],\vec{a})\|$

where $f$ is a functional defined in ${\mathcal M}$, (choice principle in ${\mathcal M}$).

Thus we have $\|\pt x\,F(x,\vec{a})\|\subset\bigcup_{\varpi\in\Pi}\|F(f[\varpi,\vec{a}],\vec{a})\|$
hence the result.
\end{proof}

\begin{lem}\label{eps-trans}
If $X\in{\mathcal M}$ is transitive, then $\gl X$ is $\eps$-transitive, i.e.:

$\,\III\,\force\pt x^{\gl X}\pt y(y\neps\gl X\to y\neps x)$.
\end{lem}
\begin{proof}
Let $x\in X$, $\varpi\in\|y\neps x\|$, i.e.\ $(y,\varpi)\in x$ and $\xi\in\LLbd$ such that $\xi\force y\neps\gl X$.
Since $X$ is transitive, we have $y\in X$ and therefore $\xi\force\bot$.
\end{proof}

\begin{thm}\label{V_alpha}
Let $\mathfrak{L}$ be the language $\{\veps,\in,\subset\}$ of \ZFe, with a symbol for each
functional definable in ${\mathcal M}$. Then, there exists an $\veps$-transitive  $\mathfrak{L}$-elementary substructure
$\wt{\mathcal N}$ of ${\mathcal N}$ such that, for all $a$ in $\wt{\mathcal N}$, there is an ordinal $\alpha$
of ${\mathcal M}$ such that $\wt{\mathcal N}\models a\eps\gl V_\alpha$.
\end{thm}
\begin{proof}
$\wt{\mathcal N}$ is made up of the elements $a$ of ${\mathcal N}$ such that $a\eps\gl V_\alpha$ for some ordinal $\alpha$
of ${\mathcal M}$ (note that it is not a class defined in ${\mathcal N}$). By Lemma~\ref{eps-trans}, each $\gl V_\alpha$ is
$\veps$-transitive and therefore $\wt{\mathcal N}$ is also.

Let $F(x,\vec{y})$ be a formula of $\mathfrak{L}$ and $\vec{b}=(b_1,\ldots,b_k$) be in $\wt{\mathcal N}$. 
We assume $\wt{\mathcal N}\models\pt x\,F(x,\vec{b})$ and we have to show that ${\mathcal N}\models\pt x\,F(x,\vec{b})$
which we do by induction on $F$. By Lemma~\ref{choix_pi}, it suffices to show that
${\mathcal N}\models\pt\varpi^{\gl\Pi}F(f[\varpi,\vec{b}],\vec{b})$.
Thus let $\pi\eps\gl\Pi$; by definition of $\wt{\mathcal N}$, there exists $X$ in ${\mathcal M}$ such that $\vec{b}\eps\gl X$.
Now, there exists $V_\alpha$ in ${\mathcal M}$ such that  $f:\Pi\fois X\to V_\alpha$ and therefore, in ${\mathcal N}$,
we have $f:\gl\Pi\fois\gl X\to\gl V_\alpha$. It follows that $f[\pi,\vec{b}]\eps\gl V_\alpha$ and therefore
$f[\pi,\vec{b}]$ is in $\wt{\mathcal N}$.

Thus $\wt{\mathcal N}\models F(f[\pi,\vec{b}],\vec{b})$, hence ${\mathcal N}\models F(f[\pi,\vec{b}],\vec{b})$
by the induction hypothesis.
\end{proof}
Replacing ${\mathcal N}$ by this elementary substructure, we shall suppose from now on:

\vspace{1ex}
\cercle{1}\hspace{2em}\emph{For all $a$ in ${\mathcal N}$, there is an ordinal $\alpha$ of ${\mathcal M}$ such that
${\mathcal N}\models a\eps\gl V_\alpha$.}

\bigskip\noindent
Theorem~\ref{aNprcl} below is not really useful in the following, but it gives a welcome information on the
(very complex) structure of the r.m.\ ${\mathcal N}\!.$

\begin{thm}\label{aNprcl}
$a{\mathcal N}$ is a proper class for all $a\eps\gl2,a\ne0$; i.e.:

$\III\,\force\pt x(\pt a^{\gl2}(\pt y(ay\eps x)\to a=0))$.
\end{thm}
\begin{proof}
Let us show, in fact, that $\,\III\force\pt x(\pt a^{\gl2}(ax\eps x\to a=0))$; in other words:

$\III\force0\eps x\to0=0$ and $\III\force x\eps x\to 1=0$.

We have $(0\eps x)\equiv(0\neps x\to\bot)$ and $(0=0)\equiv(\bot\to\bot)$ hence the first result.

Furthermore, we have $\|x\neps x\|=\{\pi\in\Pi\;;\;(x,\pi)\in x\}=\vide$; thus:

$\|x\eps x\|=\|x\neps x\to\bot\|=\|\top\to\bot\|$ and $\|1=0\|=\|\top\to\bot\|$.

Hence the second result.
\end{proof}

\section{Extensional generic extensions}\label{e_ge}
In this section we build some tools in order to manage generic extensions ${\mathcal N}_\in[G]$ of the
\emph{extensional} model ${\mathcal N}_\in$. We define a new r.a.\ and give, in this r.a., a new way to compute
the truth value of ZF-formulas in ${\mathcal N}_\in[G]$.

Let $\VV$ be fixed in ${\mathcal M}$. We have $\|x\neps\gl\VV\|=\|\lbr x\in\VV\rbr\ne1\|$ and therefore:

\centerline{$\III\,\force\pt x(x\eps\gl\VV\dbfl\lbr x\in\VV\rbr=1)$.}
In other words, the $\veps$-elements of $\gl\VV$ are exactly the elements of $\VV$ in the Boolean
model~${\mathcal M}_{\gl2}$. In fact, the formula $\pt x^{\gl\mbox{\scriptsize\bf V}}F(x)$ is the same as
$\pt x(\lbr x\in\VV\rbr=1\hto F(x))$.

Remember also the important (and obvious) equivalence:

\centerline{$\III\force\pt x\pt y(x=y\dbfl\lbr x=y\rbr=1)$}

which follows from $\|x\ne y\|=\|\lbr x=y\rbr\ne1\|$ and which identifies the r.m.\ ${\mathcal N}$ with
the Boolean model ${\mathcal M}_{\gl2}$ (Remark~\ref{mod_N}).

\begin{thm}\label{e_par_fin}
${\mathcal N}\models\pt a\pt n\indi\left(a\eps(\gl\VV)^n\to
\ex!b^{\gl(\mbox{{\scriptsize\bf V}}^n)}\pt i\indi(i<n\to a(i)=b[i])\right)$.

In other words, each finite sequence of $\;\gl\VV$ in ${\mathcal N}$ is represented by a unique finite sequence
of~$\VV$ in the boolean model ${\mathcal M}_{\gl2}$.
\end{thm}
\begin{proof}
Unicity. Note first that, since there is no extensionality in ${\mathcal N}$, you may have two sequences
$a\ne a'\eps(\gl\VV)^n$ such that $a(i)=a'(i)$ for $i<n$. Now suppose $b,b'\eps\gl(\VV^n)$ be such
that $b[i]=b'[i]$ for $i<n$. Then, we have $\lbr b,b'\in\VV^n\rbr=1$ and $\lbr(\pt i<n)(b[i]=b'[i])\rbr=1$.
Since the boolean model ${\mathcal M}_{\gl2}$ satisfies extensionality, we get $\lbr b=b'\rbr=1$ that is $b=b'$.

\smallskip\noindent
Existence. Proof by induction on $n$. This is trivial if $n=0$: take $b=\vide$.

Now, let $a\eps(\gl\VV)^{n+1}$ and $a'$ be a restriction of $a$ to $n$. Let $b'\eps\gl(\VV^n)$ such that $a'(i)=b'[i]$
for $i<n$ (induction hypothesis).

In the ground model ${\mathcal M}$, we define the binary functional $+$ as follows:

if $u$ is a finite sequence $(u_0,\ldots,u_{n-1})$, then $u+v$ is the sequence $(u_0,\ldots,u_{n-1},v)$.

We extend it to ${\mathcal N}$ and we set $b=b'+a(n)$ i.e.\ $\lbr b=b'+a(n)\rbr=1$. Therefore
$\lbr b[i]=b'[i]\rbr=1$ for $i<n$ and $\lbr b[n]=a(n)\rbr=1$, i.e.\ $a(i)=b[i]$ for $i<n$ and $a(n)=b[n]$.
\end{proof}
Consider an arbitrary ordered set $(C,\le)$ in the model ${\mathcal N}$. By Theorem~\ref{V_alpha} and
property~\cercle{1} (Section~\ref{sect_rm}), we may suppose that $(C,\le)\eps\gl\VV$ with $\VV=V_\alpha$ in ${\mathcal M}$.

$\gl\VV$ is $\veps$-transitive, by Lemma~\ref{eps-trans}.

As a set of forcing conditions, $C$~is equivalent to the set $\CF$ of finite subsets $X$ of $\gl\VV$ such that $X\cap C$
has a lower bound in $C$, $\CF$~being ordered by inclusion.

Thus, we can define $\CF$ by the following formula of \ZFe:

\centerline{\emph{$\CF(u)\equiv u\eps(\gl\VV)^{<\omega}\land(\mbox{Im}(u)\cap C)$ has a lower bound in $C$}}

where Im$(u)\subseteq\gl\VV$ is the (finite) image of the finite sequence $u$.

We have ${\mathcal N}\models\CF(u)\to u\eps(\gl\VV)^{<\omega}$ and therefore
${\mathcal N}\models\CF(u)\to u\eps(\gl\VV^{<\omega})$ by Theorem~\ref{e_par_fin}.
Moreover, we have ${\mathcal N}\models\CF(\vide)$.

In ${\mathcal M}$, the function $(u,v)\mapsto uv$, from $(\VV^{<\omega})^2$  into $\VV^{<\omega}$,
which is the concatenation of sequences, is associative with~$\vide$ as neutral element, also denoted by $\1$ (monoïd).

This function  extends to ${\mathcal N}$ into an application of $\gl(\VV^{<\omega})^2$ into $\gl(\VV^{<\omega})$
with the same properties. Thus, we write $uvw$ for $u(vw),(uv)w$, etc.

The formula $\CF(uv)$ of \ZFe\ means that $u,v$ are two compatible finite sequences of elements of~$C$, i.e.\ the
union of their images has a lower bound in $C$.
Thus $\CF$ becomes a set of forcing conditions equivalent to $C$ by means of this compatibility relation.

This formula has the following properties:

\centerline{$\force\CF(puvq)\to\CF(pvuq),\force\CF(puvq)\to\CF(puq), \force\CF(puq)\to\CF(puuq)$.}

It will be convenient to have only one formula and to use simply the following consequence:

\smallskip
\cercle{2}\hspace{3.2em}\emph{There exists a proof-like term $\mathfrak{c}$ such that} \
$\mathfrak{c}\force\CF(pqrtuvw)\to\CF(ptruuv)$.

\smallskip
Consider now, in the ground model ${\mathcal M}$, a r.a.\ ${\mathcal A}_0$ which gives the r.m.\ ${\mathcal N}$.

We suppose to have an operation $(\pi,\tau)\mapsto\pi^\tau$ from $\Pi\fois\Lbd$ into $\Pi$ such that:

\centerline{$(\xi\ps\pi)^\tau=\xi\ps\pi^\tau$ for every $\xi,\tau\in\Lbd$ and $\pi\in\Pi$.}

and  two new combinators $\chi$~\emph{(read)} and $\chi'$ \emph{(write)} such that:

\centerline{$\chi\star\xi\ps\pi^\tau\succcc\xi\star\tau\ps\pi$ (i.e.\ $
\xi\star\tau\ps\pi\in\bbot\Fl\chi\star\xi\ps\pi^\tau\in\bbot$)}

\centerline{$\chi'\star\tau\ps\xi\ps\pi\succcc\xi\star\pi^\tau$ (i.e.\ $\xi\star\pi^\tau\in\bbot\Fl
\chi'\star\tau\ps\xi\ps\pi\in\bbot$).}

Moreover, we suppose that $\chi,\chi'$ may be used to form proof-like terms.

Intuitively, $\pi^\tau$ is obtained by putting the term $\tau$ \emph{at the end} of the stack $\pi$, in the
same way that $\tau\ps\pi$ is obtained by putting $\tau$ \emph{at the top} of $\pi$.

We define now, in the ground model ${\mathcal M}$, a new r.a.\ ${\mathcal A}_1$; its r.m.\ will be called \emph{the extension
of~${\mathcal N}$ by a  $\CF$-generic (or a $C$-generic)}.

We define the terms $\BBB^*,\CCC^*,\III^*,\KKK^*,\WWW^*,\Ccc^*$ and $\kk^*_\pi$ by the conditions:

\cercle{3}\hspace{2em}
\begin{tabular}{l}
$\BBB^*=\BBB$; $\III^*=\III$;\\
$\CCC^*\star\xi\ps\eta\ps\zeta\ps\pi^\tau\succ\xi\star\zeta\ps\eta\ps\pi^{\mathfrak{c}\tau}$;
i.e.\ $\CCC^*=\lbd x\lbd y\lbd z(\chi)\lbd t((\chi')(\mathfrak{c})t)xzy$;\\
$\KKK^*\star\xi\ps\eta\ps\pi^\tau\succ\xi\star\pi^{\mathfrak{c}\tau}$;
i.e.\ $\KKK^*=\lbd x\lbd y(\chi)\lbd t((\chi')(\mathfrak{c})t)x$;\\
$\WWW^*\star\xi\ps\eta\ps\pi^\tau\succ\xi\star\eta\ps\eta\ps\pi^{\mathfrak{c}\tau}$;
i.e.\ $\WWW^*=\lbd x\lbd y(\chi)\lbd t((\chi')(\mathfrak{c})t)xyy$;\\
$\kk^*_\pi\star\xi\ps\varpi^\tau\succ\xi\star\pi^{\mathfrak{c}\tau}$;
i.e.\ $\kk^*_\pi=\lbd x(\chi)\lbd t(\kk_\pi)((\chi')(\mathfrak{c})t)x$;\\
$\Ccc^*\star\xi\ps\pi^\tau\succ\xi\star\kk^*_\pi\ps\pi^{\mathfrak{c}\tau}$; i.e.\\
$\Ccc^*=\lbd x(\chi)\lbd t(\Ccc)\lbd k(((\chi')(\mathfrak{c})t)x)\lbd x'(\chi)\lbd t'(k)((\chi')(\mathfrak{c})t')x'$.
\end{tabular}

When checking below the axioms of r.a., the property needed for each combinator is:

\cercle{4}\hspace{2em}
\begin{tabular}{l}
for $\CCC^*$: $\mathfrak{c}\force\CF(prtv)\to\CF(ptrv)$; 
for $\KKK^*$: $\mathfrak{c}\force\CF(pqr)\to\CF(pr)$;\\
for $\WWW^*$: $\mathfrak{c}\force\CF(puv)\to\CF(puuv)$;
for $\Ccc^*$: $\mathfrak{c}\force\CF(pu)\to\CF(puu)$;\\
for $\kk^*_\pi$: $\mathfrak{c}\force\CF(rtw)\to\CF(tr)$.
\end{tabular}

We get them replacing by $\1$ some of the variables $p,q,r,t,u,v,w$ in the definition~\cercle{2}
of~$\mathfrak{c}$.

We define the r.a.\ ${\mathcal A}_1$ by setting $\LLbd=\Lbd\fois\VV^{<\omega}$; $\PPi=\Pi\fois\VV^{<\omega}$;
$\LLbd\star\PPi=(\Lbd\star\Pi)\fois\VV^{<\omega}$.

$(\xi,u)\ps(\pi,v)=(\xi\ps\pi,uv)$;

$(\xi,u)\star(\pi,v)=(\xi\star\pi,uv)$;

$(\xi,u)(\eta,v)=(\xi\eta,uv)$.

The pole $\bbbot$ of ${\mathcal A}_1$ is defined by:

\centerline{$(\xi\star\pi,u)\in\bbbot\Dbfl(\pt\tau\in\Lbd)(\tau\force\CF(u)\Fl\xi\star\pi^\tau\in\bbot)$.}

The combinators are:

$\BBb=(\BBB,\1),\CCb=(\CCC^*,\1),\IIb=(\III,\1),\KKb=(\KKK^*,\1),\WWb=(\WWW^*,\1),\ccb=(\Ccc^*,\1)$;

$\kkk_{(\pi,u)}=(\kk^*_\pi,u)$; $\mbox{\PL}_{{\mathcal A}_1}$ is
$\{(\theta,\1)\;;\;\theta\in\mbox{\PL}_{{\mathcal A}_0}\}$.

\begin{thm}
${\mathcal A}_1$ is a coherent r.a.
\end{thm}
\begin{proof}
Let us prove first that ${\mathcal A}_1$ is coherent: let $\theta\in\mbox{\PL}_{{\mathcal A}_0}$; by definition of the
formula $\CF$, there exists $\tau_0\in\mbox{\PL}_{{\mathcal A}_0}$ such that $\tau_0\force\CF(\1)$.

Since $\chi'\theta\tau_0\in\mbox{\PL}_{{\mathcal A}_0}$, there exists $\pi\in\Pi$ such that
$\chi'\theta\tau_0\star\pi\notin\bbot$, thus $\theta\star\pi^{\tau_0}\notin\bbot$, therefore
$(\theta,\1)\star(\pi,\1)\notin\bbbot$.

We check now that ${\mathcal A}_1$ is a r.a.:
\begin{itemize}
\item$(\xi,u)\star(\eta,v)\ps(\pi,w)\in\bbbot\Fl(\xi,u)(\eta,v)\star(\pi,w)\in\bbbot$.

By hypothesis, we have $(\xi\star\eta\ps\pi,uvw)\in\bbbot$, therefore:

$(\pt\tau\in\Lbd)(\tau\force\CF(uvw)\to\xi\star\eta\ps\pi^\tau\in\bbot)$ and therefore:

$(\pt\tau\in\Lbd)(\tau\force\CF(uvw)\to\xi\eta\star\pi^\tau\in\bbot)$ hence the result.
\item $(\xi,u)\star(\eta,v)(\zeta,w)\ps(\pi,z)\in\bbbot\Fl
(\BBB,\1)\star(\xi,u)\ps(\eta,v)\ps(\zeta,w)\ps(\pi,z)\in\bbbot$.

By hypothesis, we have $(\xi\star\eta\zeta\ps\pi,uvwz)\in\bbbot$ therefore:

$(\pt\tau\in\Lbd)(\tau\force\CF(uvwz)\to\xi\star\eta\zeta\ps\pi^\tau\in\bbot)$ thus:

$(\pt\tau\in\Lbd)(\tau\force\CF(uvwz)\to\BBB\star\xi\ps\eta\ps\zeta\ps\pi^\tau\in\bbot)$.
\item $(\xi,u)\star(\zeta,w)\ps(\eta,v)\ps(\pi,z)\in\bbbot\Fl
(\CCC^*,\1)\star(\xi,u)\ps(\eta,v)\ps(\zeta,w)\ps(\pi,z)\in\bbbot$.

By hypothesis, we have $(\xi\star\zeta\ps\eta\ps\pi,uwvz)\in\bbbot$, therefore:

$(\pt\tau\in\Lbd)(\mathfrak{c}\tau\force\CF(uwvz)\to\xi\star\zeta\ps\eta\ps\pi^{\mathfrak{c}\tau}\in\bbot)$
therefore, by definition~\cercle{3} of $\CCC^*$:

$(\pt\tau\in\Lbd)(\mathfrak{c}\tau\force\CF(uwvz)\to\CCC^*\star\xi\ps\eta\ps\zeta\ps\pi^\tau\in\bbot)$.

But, by the property~\cercle{4} of $\mathfrak{c}$, we have:

$\tau\force\CF(uvwz)\to\mathfrak{c}\tau\force\CF(uwvz)$, hence the result.
\item $(\xi,u)\star(\pi,v)\in\bbbot\Fl(\III,\1)\star(\xi,u)\ps(\pi,v)\in\bbbot$.

Immediate.
\item $(\xi,u)\star(\pi,w)\in\bbbot\Fl(\KKK^*,\1)\star(\xi,u)\ps(\eta,v)\ps(\pi,w)\in\bbbot$.

By hypothesis, we have $(\xi\star\pi,uw)\in\bbbot$, therefore:

$(\pt\tau\in\Lbd)(\mathfrak{c}\tau\force\CF(uw)\to\xi\star\pi^{\mathfrak{c}\tau}\in\bbot)$,
therefore by the definition~\cercle{3} of $\KKK^*$:

$(\pt\tau\in\Lbd)(\mathfrak{c}\tau\force\CF(uw)\to\KKK^*\star\xi\ps\eta\ps\pi^\tau\in\bbot)$.
But, by~\cercle{4}, we have:

$\tau\force\CF(uvw)\to\mathfrak{c}\tau\force\CF(uw)$, hence the result.
\item $(\xi,u)\star(\eta,v)\ps(\eta,v)\ps(\pi,w)\in\bbbot\Fl
(\WWW^*,\1)\star(\xi,u)\ps(\eta,v)\ps(\pi,w)\in\bbbot$.

By hypothesis, we have $(\xi\star\eta\ps\eta\ps\pi,uvvw)\in\bbbot$, therefore:

$(\pt\tau\in\Lbd)(\mathfrak{c}\tau\force\CF(uvvw)\to\xi\star\eta\ps\eta\ps\pi^{\mathfrak{c}\tau}\in\bbot)$,
thus, by the definition~\cercle{3} of $\WWW^*$:

$(\pt\tau\in\Lbd)(\mathfrak{c}\tau\force\CF(uvvw)\to\WWW^*\star\xi\ps\eta\ps\pi^\tau\in\bbot)$.
Now, by~\cercle{4}, we have:

$\tau\force\CF(uvw)\to\mathfrak{c}\tau\force\CF(uvvw)$, hence the result:
$\tau\force\CF(uvw)\to\WWW^*\star\xi\ps\eta\ps\pi^\tau\in\bbot$.
\item $(\xi,v)\star(\pi,u)\in\bbbot\Fl(\kk^*_\pi,u)\star(\xi,v)\ps(\varpi,w)\in\bbbot$.

By hypothesis, we have $(\xi\star\pi,vu)\in\bbbot$, therefore:\
$(\pt\tau\in\Lbd)(\mathfrak{c}\tau\force\CF(vu)\to\xi\star\pi^{\mathfrak{c}\tau}\in\bbot)$,
thus, by the definition~\cercle{3} of $\kk^*_\pi$:

$(\pt\tau\in\Lbd)(\mathfrak{c}\tau\force\CF(vu)\to\kk^*_\pi\star\xi\ps\varpi^\tau\in\bbot)$.
Now, by~\cercle{4}, we have:

$\tau\force\CF(uvw)\to\mathfrak{c}\tau\force\CF(vu)$, hence the result:
$\tau\force\CF(uvw)\to\kk^*_\pi\star\xi\ps\varpi^\tau\in\bbot$.
\item $(\xi,u)\star(\kk^*_\pi,v)\ps(\pi,v)\in\bbbot\Fl(\Ccc^*,\1)\star(\xi,u)\ps(\pi,v)\in\bbbot$.

By hypothesis, we have $(\xi\star\kk^*_\pi\ps\pi,uvv)\in\bbbot$, therefore:

$(\pt\tau\in\Lbd)(\mathfrak{c}\tau\force\CF(uvv)\to\xi\star\kk^*_\pi\ps\pi^{\mathfrak{c}\tau}\in\bbot)$,
thus, by the definition~\cercle{3} of $\Ccc^*$:

$(\pt\tau\in\Lbd)(\mathfrak{c}\tau\force\CF(uvv)\to\Ccc^*\star\xi\ps\pi^\tau\in\bbot)$.
But, by~\cercle{4}, we have:

$\tau\force\CF(uv)\to\mathfrak{c}\tau\force\CF(uvv)$, hence the result:
$\tau\force\CF(uv)\to\Ccc^*\star\xi\ps\pi^\tau\in\bbot$. \qedhere
\end{itemize}
\end{proof}

\subsection*{The \texorpdfstring{$\mathbf{C}$}{C}-forcing defined in
\texorpdfstring{${\mathcal N}\hspace{-0.83em}{\mathcal N}$}{N}}
Let $F(\vec{a})$  be a closed formula of the language of ZF with parameters in ${\mathcal M}$.

We define the formula $p\forcec F(\vec{a})$
(read ``$p$ {\em forces} $F(\vec{a})$'') \emph{as the formula of \ZFe\ which expresses the $C$-forcing
on ${\mathcal N}_\in$}. In this formula, the variable $p$ is restricted to $\VV^{<\omega}$.

We define a subset $\llp F(\vec{a})\rrp$ of $\PPi$ by
$\llp F(\vec{a})\rrp=\{(\pi,p)\;;\;\pi\in\|p\nforcec\neg F(\vec{a})\|\}$.

\begin{lem}\label{force_fforce}
For every formula $F(\vec{x})$ of ZF\/, there exist two proof-like terms $\pF_F,\pF'_F$ of ${\mathcal A}$ such that:

i)~~$\xi\force(p\forcec F(\vec{a}))\Fl(\pF_F\xi,p)\fforce\llp F(\vec{a})\rrp$

ii)~~$(\xi,p)\fforce\llp F(\vec{a})\rrp\Fl\pF'_F\xi\force(p\forcec F(\vec{a}))$

for every $\xi\in\Lbd$ and $p\in\VV^{<\omega}$.
\end{lem}
\begin{proof}
By a well known property of forcing \cite{TJST,kunen}, the formula $p\forcec F(\vec{x})$ is equivalent, in \ZFe,
to the formula $\pt q(\CF(pq)\to q\nforcec\neg F(\vec{x}))$.

It follows that there are two proof-like terms $\mathfrak{q}_F,\mathfrak{q}_F'$ such that:

(1)~~~~$\mathfrak{q}_F\force(p\forcec F(\vec{a}))\to\pt q(\CF(pq)\to q\nforcec\neg F(\vec{a}))$;

(2)~~~~$\mathfrak{q}'_F\force\pt q(\CF(pq)\to q\nforcec\neg F(\vec{a}))\to(p\forcec F(\vec{a}))$.

{
\renewcommand{\theenumi}{\roman{enumi}}%
    \begin{enumerate}
\item By applying (1) to the hypothesis, we obtain
$\mathfrak{q}_F\xi\force\pt q(\CF(pq)\to q\nforcec\neg F(\vec{a}))$ that is:

$\pt\pi\pt q\pt\tau(\tau\force\CF(pq),\pi\in\|q\nforcec\neg F(\vec{a})\|\to
\mathfrak{q}_F\xi\star\tau\ps\pi\in\bbot)$. It follows that:

$\pt\pi\pt q(\pi\in\|q\nforcec\neg F(\vec{a})\|\to
\pt\tau(\tau\force\CF(pq)\to(\chi)(\mathfrak{q}_F)\xi\star\pi^\tau\in\bbot))$.

Now $\pi\in\|q\nforcec\neg F(\vec{a})\|$ is the same as $(\pi,q)\in\llp F(\vec{a})\rrp$ and

$\pt\tau(\tau\force\CF(pq)\to(\chi)(\mathfrak{q}_F)\xi\star\pi^\tau\in\bbot)$ is the same as
$(\pF_F\xi,p)\star(\pi,q)\in\bbbot$

with $\pF_F=\lbd x(\chi)(\mathfrak{q}_F)x$. Thus we obtain $(\pF_F\xi,p)\fforce\llp F(\vec{a})\rrp$.

\item The hypothesis gives $\pt\pi\pt q((\pi,q)\in\llp F(\vec{a})\rrp\to(\xi\star\pi,pq)\in\bbbot)$ that is:

$\pt\pi\pt q\pt\tau(\pi\in\|q\nforcec\neg F(\vec{a})\|,\tau\force\CF(pq)\to\xi\star\pi^\tau\in\bbot)$ and therefore:

$\pt\pi\pt q\pt\tau(\tau\force\CF(pq),\pi\in\|q\nforcec\neg F(\vec{a})\|\to\chi'\xi\star\tau\ps\pi\in\bbot)$.

It follows that $\chi'\xi\force\pt q(\CF(pq)\to q\nforcec\neg F(\vec{a}))$ and therefore, by (2):

$(\mathfrak{q}_F')(\chi')\xi\force(p\forcec F(\vec{a}))$. We set $\pF'_F=\lbd x(\mathfrak{q}_F')(\chi')x$.\qedhere
\end{enumerate}}
\end{proof}

In the general theory of classical realizability, we define a truth value for the formulas of 
\ZFe\ and therefore, in particular, for the formulas of ZF\/. We will define here directly a new truth value
$\vv F(a_1,\ldots,a_n)\vv$ \emph{for a formula of ZF} with parameters in ${\mathcal M}$ for the r.a.~${\mathcal A}_1$.

To this aim, we first define the truth values $\vv a\notin b\vv,\vv a\subset b\vv$ of the
atomic formulas of~ZF; then that of $F(a_1,\ldots,a_n)$, by induction on the length of the formula.

Theorem~\ref{adeq} (adequacy theorem) remains valid (cf.~Remark~\ref{remadeq}):

$\vv a\notin b\vv=~\llp a\notin b\rrp$; \ $\vv a\subset b\vv=\llp a\subset b\rrp$;

$\vv F\to F'\vv=\{(\xi\ps\pi,pq)\;;\;(\xi,p)\fforce F,(\pi,q)\in\vv F'\vv\}$;

$\vv \pt x\,F(x,a_1,\ldots,a_n)\vv=\bigcup_a\vv F(a,a_1,\ldots,a_n)\vv$.

Of course $\fforce$ is defined by: $(\xi,p)\fforce F\Dbfl(\pt(\pi,q)\in\vv F\vv)((\xi,p)\star(\pi,q)\in\bbbot)$.

\begin{rem} Be careful, as we said before, these are not the truth values, in the r.a.\ ${\mathcal A}_1$, of $a\notin b$ and $a\subset b$ \emph{considered as formulas of \ZFe}. We seek here to define directly the $C$-generic model on ${\mathcal N}_\in$ without going through a model of \ZFe.
\end{rem}

Theorem~\ref{e_forcec-fforce} below may be considered as a generalization of the well known result about
iteration of forcing: the r.a.\ ${\mathcal A}_1$, which is a kind of product of ${\mathcal A}_0$ by $\CF$, gives
the same r.m.\ as the $\CF$-generic extension of ${\mathcal N}_\in$.

\begin{thm}\label{e_forcec-fforce}
For each closed formula $F$ of ZF with parameters in the model ${\mathcal N}$, there exist two proof-like terms
$\chi_{_F},\chi'_{_F}$, which only depend on the propositional structure of $F$, such that we have,
for every $\xi\in\Lbd$ and $p\in\VV^{<\omega}$:

$\xi\force(p\forcec F)$ \ $\Fl$ \ $(\chi_{_F}\xi,p)\fforce F$ \ and \ $(\xi,p)\fforce F$ \ $\Fl$ \ $\chi'_{_F}\xi\force(p\forcec F)$.
\end{thm}
The \emph{propositional structure} of $F$ is the propositional formula built with the connective $\to$ and only
two atoms $O_\in,O_\subset$, which is obtained from $F$ by deleting all quantifiers and by identifying all atomic formulas $t\notin u,t\subset u$ respectively with $O_\in,O_\subset$.

For instance, the propositional structure of the formula:

$\pt X(\pt x(\pt y((x,y)\notin Y\to y\notin X)\to x\notin X)\to\pt x(x\notin X))$ is:

$((O_\in\to O_\in)\to O_\in)\to O_\in$.

\begin{proof} By recurrence on the length of $F$.
\begin{itemize}
\item If $F$ is atomic, we have $F\equiv a\notin b$ or $a\subset b$. Apply Lemma~\ref{force_fforce}
with $F(\vec{a})\equiv a\notin b$ or $F(\vec{a})\equiv a\subset b$.
\item ~If $F\equiv\pt x\,F'$, then $p\forcec F\equiv\pt x(p\forcec F')$. Therefore
$\xi\force p\forcec F\equiv\pt x(\xi\force(p\forcec F'))$.

Moreover, \ $(\xi,p)\fforce F\equiv\pt x((\xi,p)\fforce F')$.

The result is immediate, from the recurrence hypothesis.
\item ~If $F\equiv F'\to F''$, we have \ $p\forcec F\equiv\pt q(q\forcec F'\to pq\forcec F'')$ and therefore:

\bigskip
\cercle{5}\hspace{2em}
$\xi\force(p\forcec F)$ $\Fl$ $\pt\eta\pt q(\eta\force(q\forcec F')\to\xi\eta\force(pq\forcec F''))$.

\bigskip
Suppose that \ $\xi\force(p\forcec F)$ \ and set \ $\chi_F=\lbd x\lbd y(\chi_{F''})(x)(\chi'_{F'})y$.

We must show \ $(\chi_F\xi,p)\fforce F'\to F''$; thus, let \ $(\eta,q)\fforce F'$ and $(\pi,r)\in\vv F''\vv$.

We must show \ $(\chi_F\xi,p)\star(\eta,q)\ps(\pi,r)\in\bbbot$ \ that is \
$(\chi_F\xi\star\eta\ps\pi,pqr)\in\bbbot$.

Thus, let \ $\tau\force\mathfrak{C}(pqr)$; we must show \ $\chi_F\xi\star\eta\ps\pi^\tau\in\bbot$ \
or else \ $\chi_F\star\xi\ps\eta\ps\pi^\tau\in\bbot$.

From the recurrence hypothesis applied to $(\eta,q)\fforce F'$, we have \ $\chi'_{F'}\eta\force(q\forcec F')$.

From \cercle{5}, we have therefore \ $(\xi)(\chi'_{F'})\eta\force(pq\forcec F'')$.

Applying again the recurrence hypothesis, we get:

$((\chi_{F''})(\xi)(\chi'_{F'})\eta,pq)\fforce F''$. But since $(\pi,r)\in\vv F''\vv$, we have:

$((\chi_{F''})(\xi)(\chi'_{F'})\eta,pq)\star(\pi,r)\in\bbbot$, \ that is \
$((\chi_{F''})(\xi)(\chi'_{F'})\eta\star\pi,pqr)\in\bbbot$.

Since \ $\tau\force\mathfrak{C}(pqr)$, we have \ $(\chi_{F''})(\xi)(\chi'_{F'})\eta\star\pi^\tau\in\bbot$.

But, by definition of $\chi_F$, we have
$\chi_F\star\xi\ps\eta\ps\pi^\tau\succ(\chi_{F''})(\xi)(\chi'_{F'})\eta\star\pi^\tau$

which gives the desired result: \ $\chi_F\star\xi\ps\eta\ps\pi^\tau\in\bbot$.

Suppose now that \ $(\xi,p)\fforce F'\to F''$; we set \
$\chi'_F=\lbd x\lbd y(\chi'_{F''})(x)(\chi_{F'})y$.

We must show \ $\chi'_F\xi\force(p\forcec F'\to F'')$ \ that is \
$\pt q(\chi'_F\xi\force(q\forcec F'\to pq\forcec F''))$.

Thus, let \ $\eta\force q\forcec F'$ and $\pi\in\|pq\forcec F''\|$; we must show \
$\chi'_F\xi\star\eta\ps\pi\in\bbot$.

By the  recurrence hypothesis, we have \ $(\chi_{F'}\eta,q)\fforce F'$, therefore $(\xi,p)(\chi_{F'}\eta,q)\fforce F''$

or else, by definition of the algebra ${\mathcal A}_1$: \
$((\xi)(\chi_{F'})\eta,pq)\fforce F''$.

Applying again the recurrence hypothesis, we have \
$(\chi'_{F''})(\xi)(\chi_{F'})\eta\force(pq\forcec F'')$

and therefore $(\chi'_{F''})(\xi)(\chi_{F'})\eta\star\pi\in\bbot$. But we have, by definition of $\chi'_F$:

$\chi'_F\xi\star\eta\ps\pi\succ\chi'_F\star\xi\ps\eta\ps\pi\succ
(\chi'_{F''})(\xi)(\chi_{F'})\eta\star\pi$;

the desired result $\chi'_F\xi\star\eta\ps\pi\in\bbot$ follows. \qedhere
\end{itemize}
\end{proof}

\begin{thm}\label{theta_A}
For each axiom $A$ of ZF\/, there exists a proof like term $\Theta\!_A$ of the r.a.\ ${\mathcal A}_0$ 
such that $(\Theta\!_A,\1)\fforce A$.
\end{thm}
\begin{proof}
Indeed, if we denote by ${\mathcal N}_\in[G]$ the $C$-generic model over ${\mathcal N}_\in$, with $G\subseteq C$ being
the generic set, we have ${\mathcal N}_\in[G]\models$ ZF\/. Therefore, ${\mathcal N}\models(\1\forcec A)$, which
means that there is a proof-like term~$\Theta'\!\!_A$ such that $\Theta'\!\!_A\force(\1\forcec A)$.
By Theorem~\ref{e_forcec-fforce}, we can take $\Theta\!_A=\chi_A\Theta'\!\!_A$.
\end{proof}

\section{The algebra \texorpdfstring{$\mathfrak{A}_0$}{A0}}\label{alg_A0}
We define a r.a.\ $\mathfrak{A}_0$ which gives a very interesting r.m.\ ${\mathcal N}$.
In the following, we use only this r.a.\ and a generic extension $\mathfrak{A}_1$.

The terms of $\mathfrak{A}_0$ are finite sequences of symbols:

\centerline{$),(,\BBB,\CCC,\III,\KKK,\WWW,\Ccc,\aaa,\pp,\gamma,\kappa,\ee,\chi,\chi',h_0,h_1,\ldots,h_i,\ldots$}

$\Lbd$ is the least set which contains these symbols (except parentheses) and is such that:

$t,u\in\Lbd\Fl(t)u\in\Lbd$.

A stack is a finite sequence of terms, separated by the symbol $\ps\,$ and terminated by the symbol~$\pi_0$
\emph{(the empty stack)}.

$\Pi$ is therefore the least set such that $\pi_0\in\Pi$ and $t\in\Lbd,\pi\in\Pi\Fl t\ps\pi\in\Pi$.

$\kk_\pi$ is defined by recurrence: $\kk_{\pi_0}=\aaa$; $\kk_{t.\pi}=\lbd x(\kk_\pi)(x)t=((\CCC)(\BBB)\kk_\pi)t$.

The application $(\tau,\pi)\mapsto\pi^\tau$ from $\Lbd\fois\Pi$ into $\Pi$ consists in replacing $\pi_0$
by $\tau\ps\pi_0$.

It is therefore recursively defined by: $\pi_0^\tau=\tau\ps\pi_0$;
$(t\ps\pi)^\tau=t\ps\pi^\tau$.

$\Lbd\star\Pi$ is $\Lbd\fois\Pi$.

$\bbot$ is the least subset of $\Lbd\star\Pi$ satisfying the conditions:

\begin{enumerate}
\item $\pp\star\pi\in\bbot$ for every stack $\pi\in\Pi$ \emph{(stop)};
\item $\xi\star\pi_0\in\bbot\Fl\aaa\star\xi\ps\pi\in\bbot$ for every $\xi\in\Lbd,\pi\in\Pi$ \emph{(abort)};
\item If at least two out of $\xi\star\pi,\eta\star\pi,\zeta\star\pi$ are in $\bbot$, then
$\gamma\star\xi\ps\eta\ps\zeta\ps\pi\in\bbot$ \emph{(fork)};
\item $\xi\star\pi\in\bbot\Fl\ee\star h_i\ps h_i\ps\eta\ps\xi\ps\pi\in\bbot$ for every $\xi,\eta\in\Lbd$
and $i\in\NN$ \emph{(elimination of constants)};
\item $\xi\star\pi\in\bbot\Fl\ee\star h_i\ps h_j\ps\xi\ps\eta\ps\pi\in\bbot$ for every $\xi,\eta\in\Lbd$ and
$i,j\in\NN,i\ne j$ \emph{(elimination of constants)};
\item $\xi\star h_n\ps\pi\in\bbot\Fl\kappa\star\xi\ps\pi\in\bbot$ if $h_n$ does not appear in $\xi,\pi$
\emph{(introduction of constants)};
\end{enumerate}
and also the general axiomatic conditions for $\BBB,\CCC,\III,\KKK,\WWW,\Ccc,\chi,\chi'$ and the application:
\begin{enumerate}
    \setcounter{enumi}{6}
\item $\xi\star\eta\ps\pi\in\bbot\Fl(\xi)\eta\star\pi\in\bbot$ \emph{(push)};
\item $\xi\star\pi\in\bbot\Fl\III\star\xi\ps\pi\in\bbot$ \emph{(no operation)};
\item $\xi\star\pi\in\bbot\Fl\KKK\star\xi\ps\eta\ps\pi\in\bbot$ \emph{(delete)};
\item $\xi\star\eta\ps\eta\ps\pi\in\bbot\Fl\WWW\star\xi\ps\eta\ps\pi\in\bbot$ \emph{(copy)};
\item $\xi\star\zeta\ps\eta\ps\pi\in\bbot\Fl\CCC\star\xi\ps\eta\ps\zeta\ps\pi\in\bbot$
\emph{(switch)};
\item $\xi\star(\eta)\zeta\ps\pi\in\bbot\Fl\BBB\star\xi\ps\eta\ps\zeta\ps\pi\in\bbot$ \emph{(apply)};
\item $\xi\star\kk_\pi\ps\pi\in\bbot\Fl\Ccc\star\xi\ps\pi\in\bbot$ \emph{(save the stack)};
\item $\xi\star\tau\ps\pi\in\bbot\Fl\chi\star\xi\ps\pi^\tau\in\bbot$
\emph{(read the end of the stack)};
\item $\xi\star\pi^\tau\in\bbot\Fl\chi'\star\tau\ps\xi\ps\pi\in\bbot$ \emph{(write at the end of the stack)}.
\end{enumerate}

The property

$\xi\star\pi\in\bbot\Fl\kk_\pi\star\xi\ps\varpi\in\bbot$ \emph{(restore the stack)}

now follows easily from the definition of $\kk_\pi$.

A term is defined to be \emph{proof-like} if it contains neither $\aaa,\pp$ nor any $h_i$.

If $\xi\in$~\PL\ (the set of proof-like terms), then the process $\xi\star\pi_0$ does not contain the symbol~$\pp$;
thus $\xi\star\pi_0\notin\bbot$. It follows that \emph{the algebra $\mathfrak{A}_0$ is coherent.}

\begin{thm}
The Boolean algebra $\gl2$ has at most 4 elements.
\end{thm}
\begin{proof}
Let us show that $\gamma\force\pt x^{\gl2}\pt y^{\gl2}(xy\ne0,y\ne1,xy\ne y\to\bot)$.
This means exactly that $\gamma$ realizes the three formulas
$(\bot,\bot,\top\to\bot),\,(\bot,\top,\bot\to\bot)$ and $(\top,\bot,\bot\to\bot)$, which follows
immediately from rule~3 of the definition of $\bbot$.
\end{proof}
If $\gl2$ is trivial, everything in the following is also (cf.~Remark~\ref{rem_theta_AC}).
Therefore, we now assume that $\gl2$ has 4 elements.

\smallskip\noindent
Let $a_0,a_1$ be the atoms of $\gl2$. Then ${\mathcal M}_{a_0}=a_0{\mathcal N}$ and
${\mathcal M}_{a_1}=a_1{\mathcal N}$ are classes in the r.m.~${\mathcal N}$ respectively defined by the formulas $x=a_0x$
and $x=a_1x$.

We define the binary functional $\sqcup$ in ${\mathcal N}$ as the extension of the functional $(x,y)\mapsto x\cup y$
on~${\mathcal M}$. We do not use the symbol $\cup$, because it already denotes the union. For instance, we have
$\gl\{0\}\sqcup\gl\{1\}=\gl2$ but $\gl\{0\}\cup\gl\{1\}=\{0,1\}$.

The identity $x=a_0x\sqcup a_1x$ gives a bijection from ${\mathcal N}$ onto ${\mathcal M}_{a_0}\fois{\mathcal M}_{a_1}$.

We have ${\mathcal M}\prec{\mathcal M}_{a_0},{\mathcal M}_{a_1}$. Let us show that one of them, say ${\mathcal M}_{a_0}$, is well founded
in ${\mathcal N}$ and therefore:

\centerline{\em ${\mathcal M}_{a_0}={\mathcal M}_{\mathcal D}$ and the class of ordinals {\rm On} is defined
in ${\mathcal N}$.}

\begin{lem}\label{wellf}
The relation $\lbr x\in y\rbr=1$ is well founded.
\end{lem}
\begin{proof}
We show $\Y\force\pt y(\pt x(\lbr x\in y\rbr=1\hto F(x))\to F(y))\to\pt y\,F(y)$ for any formula $F$ of \ZFe.
Let $\xi\force\pt y(\pt x(\lbr x\in y\rbr=1\hto F(x))\to F(y))$. We show, by induction on
$\mbox{rk}(y_0)$ that $\Y\xi\force F(y_0)$. Thus we have $(\pt x\in y_0)(\Y\xi\force F(x))$ i.e.
$\Y\xi\force\pt x(\lbr x\in y_0\rbr=1\hto F(x))$ and therefore $(\xi)(\Y)\xi\force F(y_0)$ by hypothesis on $\xi$.
Hence the result since $\Y\xi\succ(\xi)(\Y)\xi$.
\end{proof}\noindent
Now, the relation $\lbr x\in y\rbr=1$ is the product of the relations $\lbr a_0x\in a_0y\rbr=a_0$ and
$\lbr a_1x\in a_1y\rbr=a_1$ respectively defined on $a_0{\mathcal N}$ and $a_1{\mathcal N}$. Since their product is
well founded, one of them must be.

\begin{rem} If the ground model ${\mathcal M}$ satisfies V = L, then ${\mathcal M}_{a_0}={\mathcal M_{\mathcal D}}$
is isomorphic to L$^{{\mathcal N}_\in}$, the class of constructible sets of ${\mathcal N}_\in$.
\end{rem}

\subsection*{Execution of processes}
If a given process $\xi\star\pi$ is in $\bbot$, it is obtained by applying precisely one of the rules of definition of $\bbot$:
if $\xi$ is an application, it is rule~7, else $\xi$ is an instruction and there is one and only one corresponding rule.

In this way, we get a finite tree, which is the proof that $\xi\star\pi\in\bbot$; it is linear, except in the case of rule~3
(instruction $\gamma$) where there is a triple branch. This tree is called \emph{the execution of the process $\xi\star\pi$}.

We can, of course, build this tree for any process; it may then be infinite.

\subsection*{Computing with the instruction \texorpdfstring{$\gamma\hspace{-0.47em}\gamma$}{gamma}}
Let us consider, in the ground model ${\mathcal M}$, some functions $f_i:\NN^{k_i}\to\NN$, satisfying a set of
axioms of the form:

$(\pt\vec{x}\in\NN^k)(t_0[\vec{x}]=u_0[\vec{x}],\ldots,t_{n-1}[\vec{x}]=u_{n-1}[\vec{x}]\to t[\vec{x}]=u[\vec{x}])$ \newline
where $t_i,u_i$ are terms built with the symbols $f_i$ and the variables $\vec{x}$.

Moreover, one of these axioms concerns a particular function $f$, and says that $f$ has at most one zero:
$(\pt x,y\in\NN)(f[x]=f[y]=0\to x=y)$.

These function symbols are also interpreted in ${\mathcal N}$ where we have $f_i:\gl\NN^{k_i}\to\gl\NN$ and in
this way, we have a set ${\mathcal E}$ of realized axioms:

\centerline{$\III\force\pt\vec{x}^{\gl\NN^k}(t_0[\vec{x}]=u_0[\vec{x}],\ldots,t_{n-1}[\vec{x}]=u_{n-1}[\vec{x}]
\hto t[\vec{x}]=u[\vec{x}])$}

and in particular:

\centerline{$\III\force\pt x^{\gl\NN}\pt y^{\gl\NN}(f[x]=0,f[y]=0\hto x=y)$.}

Now suppose that, with the axioms ${\mathcal E}$ and some other axioms realized in ${\mathcal N}\!$, which do not contain
the symbols $f_i$, like for instance:

\centerline{${\mathcal E}$ + \ZFe + \emph{$\gl2$ has at most 4 elements} + \emph{$\gl\NN$ is countable + $\cdots$}}

we can prove $\ex n\indi(f[n]=0)$. Since all these axioms are realized, we get, by Theorem~\ref{adeq},
a proof-like term $\theta$ such that $\theta\force\pt n\indi(f[n]\ne0)\to\bot$ and $\theta$ may contain
the instructions $\ee,\kappa$ and, above all, $\gamma$.

We shall show how $\theta$ allows to compute the (unique) solution $n_0$ of the equation $f[n]=0$.

Note that \emph{we do not assume that the $f_i$ (in particular $f$) are recursive} (in fact, it is not even necessary
that their domain or range is $\NN$). On the other hand, we may add symbols for all recursive  functions,
since they are defined by axiom systems of this form.

Let us add a term constant $\delta$ and the rule \
$\delta\star\ul{n}_0\ps\pi\in\bbot$ in the definition of~$\bbot$ (of course, without knowing the actual value of $n_0$).
Thus, we have $\delta\force\pt n\indi(f[n]\ne0)$ and therefore $\theta\star\delta\ps\pi_0\in\bbot$.

But any process $\xi\star\pi\in\bbot$ which does not contain $\pp$ (but possibly containing $\delta$)
computes~$n_0$. We show this by recurrence on the number of applications of rules for $\bbot$ used to build it:

If this number is 1, the process is $\delta\star\ul{n}_0\ps\pi$. Else, the only non trivial case is
when the process is $\gamma\star\xi\ps\eta\ps\zeta\ps\pi$ and when two out of the three
processes $\xi\star\pi,\eta\star\pi,\zeta\star\pi$ are in~$\bbot$ (but we do not know which).
By the recurrence hypothesis, at least two of them will give the integer~$n_0$ which is therefore
determined as the only integer obtained at least two times.
Example: $\gamma\star(((\gamma)(\delta)\ul{n})(\delta)\ul{n})(\delta)\ul{n}_0
\ps(((\gamma)(\delta)\ul{n}_0)(\delta)\ul{n})(\delta)\ul{n}_0
\ps(((\gamma)(\delta)\ul{n})(\delta)\ul{n}_0)(\delta)\ul{n}_0\ps\pi$.

\begin{rem}
Since $f$ is not supposed recursive, the trivial method to find $n_0$ by trying
  $0,1,2,\ldots$
may not be available. A working algorithm is by proving $f[n_0]=0$ by equational deduction from~${\mathcal E}$.
Both are, in general, very heavy or totally impracticable. Moral: it is better to use the powerful set theoretical
axioms than only the weak equational ones to get a working program.

If $f$ may have several zeroes, we can easily define $f'$ by a system of equations such that $f'$ has only one
zero which is the first zero of $f$, and apply the method to $f'$. This supposes $f$ to be recursive.
\end{rem}

\subsection*{\texorpdfstring{$\gl\hspace{-0.37em}\gl$}{gimel}2 is not (always) trivial}
Although this will not be used in the following, it is interesting to show that there is a r.m.
of $\mathfrak{A}_0$ in which $\gl2$ is not trivial and is therefore the 4-elements Boolean algebra.

\begin{lem}\label{intersec}
If $\theta\force\pt x^{\gl2}(x\ne0,x\ne1\to\bot)$ then, for any formulas $F,G$:

$\theta'\force F,G\to F$ and $\theta'\force F,G\to G$ with $\theta'=\lbd x\lbd y(\Ccc)\lbd k((\theta)(k)x)(k)y$.
\end{lem}
\begin{proof}
The hypothesis means that $\theta\force\bot,\top\to\bot$ and $\theta\force\top,\bot\to\bot$.

We have $x:F,y:G,k:\neg F\vdash kx:\bot,(\theta)(k)x:\top\to\bot,((\theta)(k)x)(k)y:\bot$.

It follows that $(\Ccc)\lbd k((\theta)(k)x)(k)y:F$ hence the first result.

We have $x:F,y:G,k:\neg G\vdash ky:\bot,(\theta)(k)x:\bot\to\bot,((\theta)(k)x)(k)y:\bot$.

It follows that $(\Ccc)\lbd k((\theta)(k)x)(k)y:G$ hence the second result.
\end{proof}
Suppose that, in every r.m.\ for the r.a.\ $\mathfrak{A}_0$, the Boolean algebra $\gl2$ is trivial.

Then the hypothesis of Lemma~\ref{intersec} is satisfied for some proof-like term $\theta$.

Choose the formulas $F\equiv\ex n\indi(n=0)$ and $G\equiv\ex n\indi(n=1)$. Then, we have
$\lbd x\,x\ul{0}\force F$ and $\lbd x\,x\ul{1}\force G$. By Lemma~\ref{intersec}, it follows
that $\theta''=((\theta')\lbd x\,x\ul{0})\lbd x\,x\ul{1}\force F$ and also $\force G$.

Then, by the algorithm given above, the proof-like term $\theta''$ computes simultaneously $0$ and~$1$,
which is a contradiction.

\subsection*{\texorpdfstring{$\gl\hspace{-0.37em}\gl\NN\hspace{-0.60em}\NN$}{N} is countable}
We define two sets, in the ground model ${\mathcal M}$:

$\HH=\{h_i\;;\;i\in\NN\}$, $\HHH=\{(h_i,h_i\ps\pi)\;;\;i\in\NN,\pi\in\Pi\}$
\newline
and also the bijection $\hh:\NN\to\HH$ such that $\hh[i]=h_i$ for every $i\in\NN$.

This bijection extends to the model ${\mathcal N}$ into a bijection $\hh:\gl\NN\to\gl\HH$. Moreover, we have
trivially $\force\HHH\subseteq\gl\HH$; in fact $\III\force\pt x(x\neps\gl\HH\to x\neps\HHH$).

\begin{lem}
$\force\pt i^{\gl\NN}\pt j^{\gl\NN}(\hh[i]\eps\HHH,\hh[j]\eps\HHH,\lbr i=j\rbr\ne0\to i=j)$.
\end{lem}
\begin{proof}
By definition of $\HHH$, we have $\|\hh[i]\neps\HHH\|=\|\{h_i\}\to\bot\|$.

Therefore, it suffices to show that $\ee\force\{h_i\},\{h_j\},\lbr i=j\rbr\ne0,i\ne j\to\bot$.

Let $t\force\lbr i=j\rbr\ne0$ and $u\force i\ne j$. We must show
$\ee\star h_i\ps h_j\ps t\ps u\ps\pi\in\bbot$, which follows immediately from
the execution rule of $\ee$.
\end{proof}\noindent
Thus we have ${\mathcal N}\models\pt i^{\gl\NN}\pt j^{\gl\NN}(\hh[i]\eps\HHH,\hh[j]\eps\HHH,ia_0=ja_0\to i=j)$.
It follows that:

\centerline{${\mathcal N}\models$ \emph{($\HHH$ is countable)}.}

Indeed, if $i\eps\gl\NN$ and $\hh[i]\eps\HHH$, then $i$ is determined by $ia_0$, which is an integer of
${\mathcal M}_{a_0}$ and therefore an integer of ${\mathcal N}$.

Define the function symbols pr$_0$, pr$_1:\NN\to\NN$ by:

$n=\mbox{pr}_1[n]+\frac{1}{2}(\mbox{pr}_0[n]+\mbox{pr}_1[n])(\mbox{pr}_0[n]+\mbox{pr}_1[n]+1)$
(bijection from $\NN^2$ onto $\NN$).

\begin{thm}\label{glnh}
i)~~$\kappa\force\pt\nu^{\gl\NN}\ex n^{\gl\NN}\{\hh[n]\eps\HHH,\nu=\mbox{pr}_1[n]\}$.

ii)~~$\force\gl\NN$ is countable.
\end{thm}
\begin{proof}
i)~Let $\nu\in\NN,\pi\in\Pi$ and $\xi\force\pt n^{\gl\NN}\{\nu=\mbox{pr}_1[n]\hto\hh[n]\neps\HHH\}$.
Therefore, we have $\xi\star h_n\ps\pi\in\bbot$ for all $n\in\NN$ such that $\nu=\mbox{pr}_1[n]$. There is an
infinity of such $n$, so that we can choose one such that $h_n$ does not appear in $\xi,\pi$. It follows that
$\kappa\star\xi\ps\pi\in\bbot$.

ii)~~Since $\hh:\gl\NN\to\gl\HH$ is a bijection, we obtain a surjection from $\HHH$ onto $\gl\NN$. It follows that~~
${\mathcal N}\models$ (\emph{$\gl\NN$ is countable}).
\end{proof}

\begin{thm}\label{NEPC}
${\mathcal N}\models$ NEPC (the non extensional principle of choice).
\end{thm}\noindent
This means that for any formula $R(x,y)$ of \ZFe, there is a binary relation $\Phi(x,y)$ such that:

$\force\pt x\pt y\pt y'(\Phi(x,y),\Phi(x,y')\to y=y')$ (functional relation);

$\force\pt x\pt y(R(x,y)\to\ex y'\{R(x,y'),\Phi(x,y')\})$ (choice).

This does not give the usual principle of choice in the model ${\mathcal N}_\in$ of ZF\/ because, even
if $R$ is compatible with the extensional equivalence $=_\in$, $\Phi$ is not necessarily so.

\begin{proof}
By Lemma~\ref{choix_pi}, we have 
$\force\pt x\pt y(R(x,y)\to\ex\varpi^{\gl\Pi}R(x,f[x,\varpi]))$ where $f$ is a functional symbol
defined in ${\mathcal M}$. Now, $\Pi$ is countable in ${\mathcal M}$, thus $\gl\Pi$ is equipotent to $\gl\NN$
and therefore countable by Theorem~\ref{glnh}. Therefore, we can define $\Phi(x,y)$ as ``$y=f[x,\varpi]$
for the first $\varpi\eps\gl\Pi$ such that $R(x,f[x,\varpi])$''.
\end{proof}
By the results of \cite{kri2}, it follows also that ${\mathcal N}\models$ (\emph{$\RR$ is not well orderable}).

Note that these results do not use the instruction $\gamma$ and therefore are valid for any $\gl2$.

On the other hand, the following result uses $\gamma$, i.e.\ the fact that $\gl2$ is finite:
\begin{thm}\label{WOC}
${\mathcal N}\models$ the axiom of well ordered choice (WOC) i.e.: The product of a family
of non void sets indexed by a well ordered set is non void.
\end{thm}
\begin{proof}
This follows immediately from NEPC and the fact that On is isomorphic to a class of ${\mathcal N}$.
\end{proof}
\begin{rem} In fact, since ${\mathcal N}$ satisfies NEPC, it also satisfies the well ordered principle
of choice (WOPC).
\end{rem}

Theorem~\ref{WOC} has two interesting consequences:

\begin{enumerate}
  \item There exists a proof-like term $\Theta_{WOC}\force{\,}$WOC which means that we now have a program
for the axiom WOC, which is a $\lbd$-term with the instructions $\gamma,\kappa,\ee$.
\item This gives a new proof that AC (and even \emph{``$\,\RR$ is well orderable''}) is not a consequence
of ZF +WOC~\cite{TJAC}.
\end{enumerate}

\section{The algebra \texorpdfstring{$\mathfrak{A}_1$}{A1} and a program for AC (and others)}\label{alg_A1}

\begin{lem}\label{den}
Let $R(x,y)$ be a formula of \ZFe\ such that $R(a_0x,a_1y)$ defines, in ${\mathcal N}$, a~functional from
${\mathcal M}_{a_0}$ into~${\mathcal M}_{a_1}$ or from ${\mathcal M}_{a_1}$ into~${\mathcal M}_{a_0}$.
Then, this functional has a countable image.
\end{lem}
Remember that ${\mathcal M}_{a_i}(i=0,1)$ is the class defined by $a_ix=x$.
\begin{proof}
Suppose, for instance, that $R(x,y)$ defines a functional from ${\mathcal M}_{a_0}$ into ${\mathcal M}_{a_1}$ i.e.:

${\mathcal N}\models\pt x\pt y\pt y'(R(a_0x,a_1y),R(a_0x,a_1y')\to a_1y=a_1y')$.

Applying Lemma~\ref{choix_pi} to the formula $\neg R(a_0x,a_1y)$, we obtain:

\centerline{$\III\,\force\pt\varpi^{\gl\Pi}\neg R(a_0x,a_1f[a_0x,\varpi])\to\pt y\neg R(a_0x,a_1y)$}

for some functional $f:{\mathcal M}\fois\Pi\to{\mathcal M}$ \emph{defined in ${\mathcal M}$}.

By Lemma~\ref{fax-afx}, we have $a_1f[a_0x,\varpi]=a_1f[a_1a_0x,\varpi]$. Since $a_1a_0=0$, we get:

\centerline{$\III\,\force\pt\varpi^{\gl\Pi}\neg R(a_0x,a_1f[\vide,\varpi])\to\pt y\neg R(a_0x,a_1y)$.}

Now $\Pi$ is countable (in ${\mathcal M}$), thus $\gl\Pi$ is equipotent to $\gl\NN$; therefore $\gl\Pi$ is countable
(in~${\mathcal N}$) by Theorem~\ref{glnh}. Hence we have, ~for some surjection $g$ from $\NN$ onto $\{f[\vide,\varpi]\;;\;\varpi\eps\gl\Pi\}$:

\centerline{$\force\pt n\indi\neg R(a_0x,a_1g(n))\to\pt y\neg R(a_0x,a_1y)$.}

Thus, we have $\force\ex y\,R(a_0x,a_1y)\to\ex n\indi R(a_0x,a_1g(n))$. It follows that $a_1g$
is a surjection from~$\NN$ onto the image of $R$.
\end{proof}
\begin{rem}\label{notwf} This shows that ${\mathcal M}_{a_0}$ and ${\mathcal M}_{a_1}$ cannot be
both well founded: otherwise, their classes of ordinals would be isomorphic, which is excluded
by Lemma~\ref{den} and Theorem~\ref{aNprcl}.
Much more general results are given in~\cite{kri2,kri4}.
\end{rem}

In order to simplify a little, we suppose that the ground model ${\mathcal M}$ satisfy V = L. This is not
really necessary: the important point is the \emph{principle of choice (PC): there is a bijective functional
between ${\mathcal M}$ and} On.

We denote by On$_{a_0}$ the class of ordinals of ${\mathcal M}_{a_0}$, which is order isomorphic to
the class On of ordinals of~${\mathcal N}_\in$.

\begin{lem}\label{borne_On}
${\mathcal N}\models\ex X\pt x\ex y\{x=_\in y,\lbr a_1y\in a_1X\rbr\ge a_1\}$.
\end{lem}
\begin{proof}
For each $\alpha$ in On$_{a_0}$, we can choose, by WOPC, an element $W_\alpha$ of the class of sets
extensionally equivalent to $V_\alpha$.
Now, the functional $\alpha\mapsto a_1W_\alpha$ from On$_{a_0}$ into ${\mathcal M}_{a_1}$ has a countable image
by Lemma~\ref{den}. It follows that there exists $\alpha_0$ such that $a_1W_\alpha=a_1W_{\alpha_0}$
for an unbounded class~$U$ of $\alpha$ in On$_{a_0}$.

Now, we have $\pt x\ex\alpha\{U(\alpha),x\in W_\alpha\}$ (in any model of ZF\/, every set belongs to some
$V_\alpha$) and therefore $\pt x\ex y\ex\alpha\{U(\alpha),x=_\in y,y\eps W_\alpha\}$.
By Lemma~\ref{eps-Cl}, it follows that:

$\pt x\ex y\ex\alpha\{U(\alpha),x=_\in y,\lbr y\in\mbox{Cl}[W_\alpha]\rbr=1\}$.
But, by Lemma~\ref{fax-afx}, we have:

$a_1=a_1\lbr y\in\mbox{Cl}[W_\alpha]\rbr=a_1\lbr a_1y\in a_1\mbox{Cl}[a_1W_\alpha]\rbr$ and, by $U(\alpha)$,
$a_1W_\alpha=a_1W_{\alpha_0}$.

Therefore, we have
$\pt x\ex y\ex\alpha\{U(\alpha),x=_\in y,a_1\lbr a_1y\in a_1\mbox{Cl}[a_1W_{\alpha_0}]\rbr=a_1\}$.

Hence the result, with $X=\mbox{Cl}[a_1W_{\alpha_0}]$.
\end{proof}

\begin{thm}\label{gen_ac}
There exists a generic extension ${\mathcal N}_\in[G]$ of ${\mathcal N}_\in$ which satisfies AC.
\end{thm}
\begin{proof}
Let $Y=\{a_1y\;;\;\lbr a_1y\in a_1X\rbr\ge a_1\}$: such a set exists in ${\mathcal N}$ by Lemma~\ref{ax-in-aX-gea}.

By Lemma~\ref{borne_On}, $(a_0x,a_1y)\mapsto a_0x\sqcup a_1y$ is a surjective functional from
${\mathcal M}_{a_0}\fois Y$ onto the whole model ${\mathcal N}_\in$ of ZF.

Since ${\mathcal M}_{a_0}\models$ V = L, there exists a surjective functional $\Psi:\mbox{On}\fois Y\to{\mathcal N}_\in$
and therefore ${\mathcal N}_\in$ is the union of the $Z_\alpha$ with $\alpha$ in On, where $Z_\alpha$ is the image by $\Psi$
of $\{\alpha\}\fois Y$.

Let us consider, for each ordinal $\alpha$ of ${\mathcal N}_\in$, the equivalence relation $\simeq_\alpha$ on $Y$ defined
by $y\simeq_\alpha y'\Dbfl\Psi(\alpha,y)=_\in\Psi(\alpha,y')$. These equivalence relations form a set
(included in ${\mathcal P}(Y^2)$). Thus, there exist an ordinal $\alpha_0$ and for all $\beta$, a surjection
$S_\beta:Z\to Z_\beta$ with $Z=\bigcup_{\alpha<\alpha_0}Z_\alpha$.

Finally, using NEPC (non extensional principle of choice) we get a surjective functional from $\mbox{On}\fois Z$
onto ${\mathcal N}_\in$ (note that, by definition, $Z$ is a set of ${\mathcal N_\in}$).

Let us now make $Z$ countable (or even only well ordered) by means of a generic $G$ on~${\mathcal N}_\in$;
then ${\mathcal N}_\in[G]\models$ AC.
\end{proof}
In the following, we take for $(C,\le)$ the set of conditions of ${\mathcal N}_\in$ given by
Theorem~\ref{gen_ac}. Applying the constructions of Section~\ref{e_ge}, we obtain a r.a.
$\mathfrak{A}_1$ and a generic model ${\mathcal N}_\in[G]$ which satisfies~AC by Theorem~\ref{gen_ac}.
Therefore, we have ${\mathcal N}_\in\models(\1\forcec\mbox{AC})$.

\begin{rem}\label{rem_theta_AC}
The formula $\CF(u)$ which defines the forcing must specify that, if $\gl2$
is the 2-elements Boolean algebra, the set of forcing conditions is trivial (for instance a singleton).
Indeed, in this case, ${\mathcal N}={\mathcal M}_{\mathcal D}$ is well ordered, therefore ${\mathcal N}_\in\models$ AC, and
there is no need to extend it.
\end{rem}

It follows that $\force(\1\forcec\mbox{AC})$ and finally $\fforce$AC by Theorem~\ref{e_forcec-fforce}.
Hence the:
\begin{thm}\label{theta_AC}
There exists a proof-like term $\Theta\!_{AC}$ of the r.a.\ $\mathfrak{A}_0$ such that
$(\Theta\!_{AC},\1)\fforce\mbox{AC}$.
\end{thm}\noindent
More generally by Theorem~\ref{theta_A}, it follows that for each axiom $A$ of ZFC, there exists a proof-like term
$\Theta\!_A$ of the r.a.\ $\mathfrak{A}_0$ such that $(\Theta\!_A,\1)\fforce A$. Note that $\Theta\!_{AC}$
is the only one which contains the instructions $\gamma,\kappa,\ee$.

\subsection*{Example of computation with \texorpdfstring{$\Theta\!_{AC}\hspace{-1.56em}\Theta\!_{AC}$}{Theta}}
Consider a function $f:\NN\to2$ such that we have a proof of $\ex n\indi(f[n]=0)$ in the theory
ZF + AC + ${\mathcal E}$ where ${\mathcal E}$ is a set of axioms of the form:

$(\pt\vec{x}\in\NN^k)(t_0[\vec{x}]=u_0[\vec{x}],\ldots,t_{n-1}[\vec{x}]=u_{n-1}[\vec{x}]
\to t[\vec{x}]=u[\vec{x}])$
\newline
(cf.~Section~\ref{alg_A0}, \emph{Computing with the instruction $\gamma$}).

We denote by ${\mathcal E}_\in$ the conjunction of the corresponding set, written in the language of~${\mathcal N}_\in$:

$(\pt\vec{x}\in\NN^k)(t_0[\vec{x}]=_\in u_0[\vec{x}],\ldots,t_{n-1}[\vec{x}]=_\in u_{n-1}[\vec{x}]
\to t[\vec{x}]=_\in u[\vec{x}])$.

This proof gives a term $\Phi$ written with the only combinators $\BBB,\CCC,\III,\KKK,\WWW,\Ccc$ such that:

$\vdash\Phi:\mbox{ZFC}_0,{\mathcal E}_\in\to(\ex n\in\NN)(f[n]=_\in0)$
\newline
for some finite conjunction ZFC$_0$ of axioms of ZFC.

Therefore, by theorems~\ref{adeq}, \ref{theta_A} and~\ref{theta_AC}, we have in the r.a.\ $\mathfrak{A}_1$:

$(\Phi^*\Theta_{ZFC_0},\1)\fforce({\mathcal E}_\in\to(\ex n\in\NN)(f[n]=_\in0))$
\newline
(remember that if $t\in\Lbd$, we obtain $t^*$ replacing $\CCC,\KKK,\WWW,\Ccc$ by
$\CCC^*,\KKK^*,\WWW^*,\Ccc^*$).

\smallskip\noindent
By Theorem~\ref{e_forcec-fforce} applied with $F\equiv{\mathcal E}_\in\to(\ex n\in\NN)(f[n]=_\in0)$, it follows that:

\centerline{$(\chi'_F)(\Phi^*)\Theta_{ZFC_0}\force\big(\1\forcec({\mathcal E}_\in\to
(\ex n\in\NN)(f[n]=_\in0))\big)$.}

Since $\forcec$ is a forcing on ${\mathcal N}_\in$, and $F$ is arithmetical, we have:

\centerline{\ZFe\ $\vdash\big(\1\forcec({\mathcal E}_\in\to(\ex n\in\NN)(f[n]=_\in0))\big)\to
\big({\mathcal E}\to\ex n\indi(f[n]=0)\big)$.}

Hence, there is a proof-like term $\Xi$ in the r.a.\ $\mathfrak{A}_0$ such that:

\centerline{$(\Xi)(\chi'_F)(\Phi^*)\Theta_{ZFC_0}\force{\mathcal E}\to\ex n\indi(f[n]=0)$.}

Now we can apply the algorithm of Section~\ref{alg_A0}, \emph{Computing with the instruction $\gamma$}.

\subsection*{Programs for other axioms}
Consider an axiom $A$ that we can prove consistent with ZFC
by forcing. Let $C_0$ be the set of conditions given by Theorem~\ref{gen_ac} and take for $(C,\le)$
the set of conditions which gives the iterated forcing for $A$ over $C_0$. Applying the constructions of
Section~\ref{e_ge},  we obtain a r.a.\ $\mathfrak{A}_1$ and a generic model ${\mathcal N}_\in[G]$ which satisfies~$A$.
It follows that $\force(\1\forcec A)$ and finally $\fforce A$ by Theorem~\ref{e_forcec-fforce}. Hence:

\begin{thm}\label{theta_AC+A}
For every formula $A$ of the language of ZF which can be proved consistent with ZFC by forcing, there exists
a proof-like term $\Theta_A$ of the r.a.\ $\mathfrak{A}_0$ such that $(\Theta_A,\1)\fforce F$.
\end{thm}\noindent
Then we can use $\Theta_A$ in computations as explained above.

\begin{rem}\label{finrem}
In the introdution, I said we consider the problem of turning proofs in ZFC into programs. But we must be more
precise because this may seem rather easy by means of classical realizability; indeed, here is an algorithm to transform a proof in ZFC of an arithmetical formula $A$ into a program:
first, transform it into a proof of ZF $\vdash A$ by restricting its quantifiers to L or HOD.
Then, get a program from this new proof using c.r.

The problem solved in the present article is more difficult and is much more interesting from the point of view of
computer science: consider the program P obtained from a proof of ZF $\vdash (AC\to A$) and,
\emph{forgetting this proof}, transform directly the program~P. This is more difficult because the program P
contains much less information than the proof which gave it.
In fact, the interesting point is that this is possible.

You may compare with the problem of modifying a program, knowing its source code or only its compiled code.
It is well known that the first situation is (very) much easier.
\end{rem}

\bibliography{biblio_AC}{}
\bibliographystyle{alpha}

\end{document}